\documentclass[aps,prxquantum,10pt,twocolumn,groupaddress,floatfix,nofootinbib,longbibliography]{revtex4-2}

\usepackage{amsmath,amssymb,amsfonts,amsthm,bm,mathtools}
\usepackage{mathrsfs}
\usepackage{booktabs}
\usepackage[mathscr]{eucal}
\usepackage{bbm}
\usepackage{braket}
\usepackage{graphicx}
\usepackage{xcolor}
\usepackage[mathlines]{lineno}
\usepackage[colorlinks=true,citecolor=blue,urlcolor=blue,linkcolor=red]{hyperref}
\newtheorem{theorem}{Theorem}
\newtheorem{lemma}{Lemma}
\newtheorem{proposition}{Proposition}
\newtheorem{corollary}{Corollary}
\newtheorem{definition}{Definition}

\newtheorem{remark}{Remark}

\newtheorem{feature}{Feature}
\newtheorem{fact}{Fact}

\usepackage{enumitem}

\newcommand{\SNBose}{S. N. Bose National Centre for Basic Sciences, Block JD, Sector III, Salt Lake, Kolkata 700 106, India.}

\begin{document}

\title{\emph{All-vs-Nothing} Operational Manifestation of Preparation Contextuality}

\author{Jayashree Karmakar}\affiliation{\SNBose}
\author{Rafiuddin Gazi}\affiliation{\SNBose}
\author{Biswadeep Chatterjee}\affiliation{\SNBose}
\author{Subhendu B Ghosh}\affiliation{\SNBose}
\author{Anandamay Das Bhowmik}\affiliation{\SNBose}
\author{Ananya Chakraborty}\affiliation{\SNBose}
\author{Manik Banik}\affiliation{\SNBose}

\begin{abstract}
\noindent Preparation contextuality is often manifested operationally through a quantitative advantage over preparation-noncontextual models in information-processing tasks, yet quantum theory typically falls short of perfect success. Here we introduce a parity-oblivious Hidden Matching task that instead yields an \emph{all-vs-nothing} manifestation of preparation contextuality. Alice encodes an $n$-bit string so that her message reveals no information about any input parity other than the two-bit parities. Bob, given a perfect matching on the input positions, must output an edge of the matching together with the parity of its two bits. A quantum protocol using a single $\lceil\log_2 n\rceil$-qubit message satisfies the parity-obliviousness constraint and succeeds with certainty. We prove that no preparation-noncontextual model achieves perfect success for any even $n\ge6$. The result holds for arbitrary ontic state spaces and assumes neither outcome determinism nor measurement noncontextuality. For $n=6$ and $n=8$, the optimal noncontextual success probabilities are $4/5$ and $3/4$, respectively, while asymptotically the optimal noncontextual success probability is $1/2+\Theta(1/\sqrt n)$. Our asymptotic upper bound follows from a Fourier-analytic sum rule derived via hypercontractivity on the Boolean hypercube. The resulting separation is qualitative rather than merely quantitative: quantum theory achieves perfect success, whereas preparation-noncontextual models cannot. We aslo show that at $n=8$ the noncontextual bound remains $3/4$ even when Bob is restricted to just suitable $4$ of the $105$ possible perfect matchings, bringing the effect within experimental reach.
\end{abstract}

\maketitle

\section{Introduction}

\noindent 
The seminal Bell--Kochen--Specker theorem shows that quantum theory cannot be reproduced by a hidden-variable model in which measurement outcomes are predetermined and independent of the context in which a measurement is performed~\cite{Bell1966,Kochen1967,Mermin1993,Brunner2014,Budroni2022}. In its traditional formulation, however, noncontextuality is a restriction on models of quantum theory, formulated for sharp projective measurements with deterministic response functions. It therefore does not, by itself, provide an operationally testable statistical criterion that can be directly confronted with experimental data, as in Bell tests~\cite{Bell1964}.

Spekkens extended the notion of noncontextuality to general operational theories~\cite{Spekkens2005}. The basic requirement is that operationally equivalent procedures---that is, procedures that are indistinguishable by all experiments allowed by the theory---must be represented identically in an underlying ontological model. This requirement applies to preparations and measurements alike and does not assume determinism. Since operational equivalences are defined entirely in terms of observable statistics, they lead to experimentally testable constraints, typically in the form of noncontextuality inequalities.

Preparation noncontextuality has played a particularly prominent role in quantum information. Spekkens \emph{et al.} showed that, in parity-oblivious multiplexing (POM), a two-party communication task, quantum theory outperforms every preparation-noncontextual theory~\cite{Spekkens2009}. Related preparation-contextuality advantages have subsequently been identified in minimum-error state discrimination~\cite{Schmid2018}. These information-processing manifestations of preparation contextuality are typically \emph{quantitative}: quantum theory achieves a success probability strictly above the preparation-noncontextual bound, but below unity. In the $2\mapsto1$ POM task, for example, the noncontextual bound is $3/4$, whereas quantum theory achieves $\cos^2(\pi/8)\simeq0.854$~\cite{Spekkens2009}. The distinction therefore appears as a finite statistical gap, reminiscent of inequality-based Bell tests such as the Clauser--Horne--Shimony--Holt (CHSH) test~\cite{Clauser1969}. This contrasts with all-vs-nothing arguments of the Greenberger--Horne--Zeilinger (GHZ) type~\cite{Greenberger1989,Greenberger1990,Mermin1990}, in which quantum theory predicts with certainty an event that is impossible in the corresponding noncontextual model.

In this work, we introduce a task that realizes an analogous all-vs-nothing separation for preparation contextuality. We call it the parity-oblivious Hidden Matching ($\mathrm{PoHM}_n$) problem which is a particularly constrained variant of Hidden Matching ($\mathrm{HM}_n$) problem. The $\mathrm{HM}_n$ task was introduced by Bar-Yossef, Jayram, and Kerenidis to establish an exponential quantum advantage over classical communication in one-way communication complexity~\cite{BarYossef2008}. Alice holds an $n$-bit string $x\in\{0,1\}^n$, while Bob is given a perfect matching $\mathtt M$ on $[n]$. Bob must output an edge $\{i,j\}\in\mathtt M$ together with the parity $x_i\oplus x_j$. In addition, parity-obliviousness constraints Alice's message to be oblivious to every parity of $x$ including the individual bit information other than the two-bit parities ${x_i\oplus x_j}$. The exemption applies to all $\binom{n}{2}$ two-bit parities, not merely those associated with $\mathtt M$, since the matching is known only to Bob and hence cannot influence Alice's encoding. Thus, the constraint excludes precisely the information that are irrelevant to the task.

Quantum theory meets both requirements simultaneously. The $\lceil\log_2 n\rceil$-qubit encoding $\ket{\psi_x}\propto\sum_i(-1)^{x_i}\ket{i}$ has a density operator that depends on $x$ only through the two-bit parities: the amplitudes are linear in the signs $(-1)^{x_i}$, so the density-matrix entries are quadratic in these signs and contain no higher-order characters. A measurement in the basis adapted to $\mathtt M$ then returns an edge of $\mathtt M$ together with its parity with certainty. No classical protocol can achieve the same. The obstruction is sharp and simple: if a message determines the pairwise parities among four coordinates, then it also determines their four-bit parity, violating the parity-obliviousness constraint. Hence, for each message, the coordinates whose pairwise parities are fixed form classes of size at most three, while a simple combinatorial argument shows that, for $n\ge6$, there exists a perfect matching with no edge contained within any such class.

Preparation noncontextuality converts this classical impossibility into a statement about ontological models. The obliviousness condition is formulated entirely in terms of operational statistics and renders certain mixed preparations operationally equivalent. Preparation noncontextuality then requires their ontic representations to coincide, implying that the ontic state carries no information about the forbidden parities. The ontic state can therefore be treated as a classical message obeying the same constraint, and the classical impossibility applies to it—regardless of the size or structure of the ontic state space, and irrespective of how much information it may carry about $x$ as a whole.

Because a possibilistic argument is fragile under experimental imperfections, we also determine how well a preparation-noncontextual theory can perform. We show that the task reduces exactly to a finite-dimensional convex optimization problem: maximizing a matching functional over nonnegative Fourier polynomials of degree two. The resulting optimal noncontextual success probability is $\tfrac12+\Theta(n^{-1/2})$, with the upper bound following from a sum rule for pairwise correlations derived from hypercontractivity on the Boolean hypercube~\cite{Bonami1970,Beckner1975,Montanaro2012,deWolf2008}. The $n^{-1/2}$ scaling coincides with that appearing in the communication-complexity bound for Hidden Matching~\cite{BarYossef2008}, although the two results constrain different resources.

We further determine the optimal noncontextual success probability, obtaining the exact values $4/5$ for $n=6$ and $3/4$ for $n=8$. A direct implementation of the task requires one measurement for each perfect matching, of which there are $(n-1)\times(n-3)\cdots3\times1\equiv(n-1)!!$. Remarkably, for $n=6$, only $3$ measurements suffice, with the corresponding optimal noncontextual success probability increasing from $4/5=0.80$ to $5/6\approx0.83$. For $n=8$, only $4$ measurements suffice, while the optimal noncontextual success probability remains $3/4$. Quantum always achieves the perfect success. These finite-size instances therefore retain a clear quantitative separation while requiring only a small number of measurement settings, making them particularly amenable to experimental tests.

The paper is organized as follows. Section~\ref{sec:ontological} reviews the ontological-model framework and introduces the notation used throughout. Section~\ref{sec:pohm} defines the \(\mathrm{PoHM}_n\) task, establishes the classical and preparation-noncontextual impossibility results, and presents the perfect quantum protocol. Section~\ref{sec:inequality} derives the quantitative preparation-noncontextuality inequality and determines its asymptotic scaling. Section~\ref{sec:discussion} discusses the implications, connection to prior works, and experimental prospects of the results.

\section{The ontological models framework}\label{sec:ontological}
\noindent 
An ontological model provides a framework for representing a physical theory in terms of underlying ontic states whose statistical predictions reproduce its operationally accessible phenomena. We briefly recall the essential ingredients of this framework; for a comprehensive account, we refer the reader to~\cite{Spekkens2005} (see also~\cite{Harrigan2007,Harrigan2010,Leifer2014}).

\begin{definition}[Operational theory]\label{def:optheory}
An operational theory specifies a set of preparation procedures $\mathcal{P}$ and a set of measurement procedures $\mathcal{M}$, each described by a collection of laboratory instructions, together with a rule that assigns to every pair $P\in\mathcal{P}$ and $M\in\mathcal{M}$ a probability distribution $p(k\mid P,M)$ over the possible outcomes $k$ of $M$. In quantum theory,
\begin{align}
p(k\mid P,M)=\operatorname{Tr}(\rho_{_P} E_M^k),
\end{align}
where $\rho_{_P}$ is the density operator associated with the preparation $P$, and $\{E_M^k\}_k$ is the positive operator valued measure (POVM) associated with the measurement $M$, satisfying $E_M^k\ge 0,~\sum_k E_M^k=\mathbf{I}$, where $\mathbf{I}$ denotes the identity operator acting on the Hilbert space associated with the system.
\end{definition}
\noindent An operational theory may also include a set of transformation procedures $\mathcal{T}$ describing physical processes that act on the system between preparation and measurement. A transformation can modify the state produced by a preparation before it is subsequently measured, thereby entering the operational statistics through the composed procedure involving the preparation, transformation, and measurement. In this work, we restrict attention to the prepare-and-measure scenario: transformations are neither included as primitive operational procedures nor used to formulate operational equivalences or noncontextuality conditions.

\begin{definition}[Ontological model]\label{def:ontmodel}
An ontological model of an operational theory associates with it a measurable space $\Lambda$ of ontic states, together with
\begin{enumerate}[itemsep=-.05cm, topsep=2pt, leftmargin=.5cm]
\item for each preparation $P\in\mathcal{P}$, a probability measure $\mu(\lambda|P)$ on $\Lambda$, and
\item for each measurement $M\in\mathcal{M}$ and outcome $k$, a measurable response function $\xi(k|\lambda,M)\in[0,1]$ satisfying $\sum_k \xi(k|\lambda,M)=1$ for every $\lambda\in\Lambda$.
\end{enumerate}
The model reproduces the operational statistics according to
\begin{align}\label{eq:ontrep}
p(k\mid P,M)=\int_\Lambda \xi(k|\lambda,M)\,\mu(\lambda|P)\,d(\lambda).
\end{align}
\end{definition}
 
\noindent An ontological model is called \emph{outcome deterministic} if
\begin{align}
\xi(k|\lambda,M)\in{0,1}\quad\forall\,k,\,M,\,\lambda\,.
\end{align}
Outcome determinism, together with measurement noncontextuality, is a central assumption in the Bell--Kochen--Specker (BKS) theorem, which rules out noncontextual assignments of pre-existing definite values to quantum observables~\cite{Kochen1967} (see also~\cite{Mermin1993,Brunner2014,Kunjwal2015,Budroni2022}). Spekkens generalized the notion of noncontextuality by removing the need for deterministic response functions and formulating noncontextuality directly in terms of operational equivalences between general preparation and measurement procedures~\cite{Spekkens2005}. This framework therefore extends the KS notion of contextuality from sharp projective measurements to general operational scenarios, while allowing intrinsically probabilistic response functions.

\begin{definition}[Operational equivalence]\label{def:opequiv}
Two preparations $P$ and $P'$ are said to be operationally equivalent, denoted by $P\simeq P'$, if they are statistically indistinguishable by every measurement allowed in the operational theory:
\begin{align}
p(k\mid P,M)=p(k\mid P',M)
\quad\forall\,k,\,M.
\end{align}
Equivalently, no measurement in the operational theory can distinguish $P$ from $P'$. 
\end{definition}

\noindent 
As a simple example in quantum theory, consider the following two distinct preparation procedures for a qubit:
\begin{align*}
\left.
\begin{aligned}
&P:\quad \text{prepare } |0\rangle \text{ or } |1\rangle
\text{ with equal probability},\\
&P':\quad \text{prepare } |+\rangle \text{ or } |-\rangle
\text{ with equal probability},
\end{aligned}\right\}
\end{align*}
where $|\pm\rangle:=(|0\rangle\pm|1\rangle)/\sqrt{2}$. The corresponding density operators are identical: $\rho_{_P}=1/2(|0\rangle\langle 0|+|1\rangle\langle 1|)=\mathbf{I}/2=1/2(|+\rangle\langle +|+\frac{1}{2}|-\rangle\langle -|)=\rho_{_{P'}}$. Consequently, for any measurement $M$ with POVM elements $\{E_M^k\}_k$,
\begin{align*}
\operatorname{Tr}(\rho_P E_M^k)=\operatorname{Tr}(\rho_{P'}E_M^k)\quad\forall,k.
\end{align*}
Thus, $P\simeq P'$; although the two preparations are implemented by physically distinct procedures, no measurement in the operational theory can distinguish them.
 
\begin{definition}[Preparation noncontextuality]\label{def:pnc}
An ontological model is said to be preparation noncontextual if operationally equivalent preparations are assigned identical ontological representations; that is,
\begin{align}\label{eq:pnc}
P\simeq P'\quad\Longrightarrow\quad\mu(\lambda|P)=\mu(\lambda|{P'})\quad\forall\; \lambda\,.
\end{align}
\end{definition}

\noindent 
The content of \eqref{eq:pnc} is that the context of a preparation---namely, the particular laboratory procedure used to implement it---leaves no imprint on its ontological representation beyond the operational statistics it produces.\footnote{This requirement may be viewed as an operational formulation of the Leibniz principle of the \emph{identity of empirical indiscernibles}~\cite{LeibnizClarke1998}: distinctions that are inaccessible at the operational level should not be introduced at the ontological level~\cite{Spekkens2019}.} The analogous condition for measurements is \emph{measurement noncontextuality}: if two measurements $M$ and $M'$ are operationally equivalent, meaning that $p(k\mid P,M)=p(k\mid P,M')~\forall\,P,k$, then they must be assigned identical response functions, $\xi(k\mid\lambda,M)=\xi(k\mid\lambda,M')~\forall\,\lambda,k$. In this work, however, we impose and analyze preparation noncontextuality only; no assumption of measurement noncontextuality is required.

One further ingredient is needed for our purpose. Unlike preparation noncontextuality, it is not an additional noncontextuality assumption, but a consequence of the standard convex structure of an ontological model.

\begin{feature}[Mixtures]\label{fea:mixture}
Let $\{P_\alpha\}$ be a collection of preparation procedures, and let $P$ denote the procedure that first samples $\alpha$ with probability $w_\alpha$, where $w_\alpha\ge0$ and $\sum_\alpha w_\alpha=1$, and then implements $P_\alpha$. Equivalently, $
P=\sum_\alpha w_\alpha P_\alpha$. The corresponding ontic distribution is given by the same convex combination,
\begin{equation}
\mu(\lambda\mid P)=\sum_\alpha w_\alpha,\mu(\lambda\mid P_\alpha).
\end{equation}
Thus, classical randomization of preparation procedures is represented by the corresponding convex mixture of their ontic distributions.
\end{feature}

\noindent 
Spekkens showed that even the maximally mixed state of a qubit is preparation contextual: distinct preparation procedures realizing the same quantum state must, under preparation noncontextuality, be assigned the same ontological representation, but this requirement is incompatible with reproducing the quantum statistics~\cite{Spekkens2005}. In particular, the maximally mixed state has a continuum of distinct convex decompositions into pure states, providing a family of operationally equivalent preparations whose ontological representations cannot all be identified consistently. 

This foundational observation established preparation contextuality as an operational feature of quantum theory and motivated its study in information-processing tasks. A canonical example is parity-oblivious multiplexing (POM), where Alice receives an $n$-bit string $x\in\{0,1\}^n$, Bob receives an index $y\in[n]$, and Alice sends a single message from which Bob must guess $x_y$. The defining constraint is that the message reveal no information about any parity involving two or more input bits; equivalently, for every $S\subseteq[n]$ with $|S|\ge2$, the message distributions conditioned on $x_S=0$ and $x_S=1$ must coincide~\cite{Spekkens2009}. Spekkens \emph{et al.} showed that preparation-noncontextual models obey a bound on the success probability of this task that is surpassed by quantum theory~\cite{Spekkens2009}. This work subsequently inspired a broader family of communication tasks in which quantum advantages can be attributed to preparation contextuality~\cite{Banik2014,Banik2015,Chailloux2016,Hameedi2017,Ghorai2018,Saha2019,Ambainis2019,Pan2019,Fonseca2025,Roy2026}. Preparation contextuality has also emerged as an operational resource in quantum state discrimination. Schmid and Spekkens derived preparation-noncontextual bounds on the success probability of minimum-error discrimination and showed that quantum strategies can surpass these bounds~\cite{Schmid2018}. Subsequent works have extended this framework to a variety of state-discrimination settings, including different state ensembles, prior-probability distributions, and discrimination criteria~\cite{Shin2021,Mukherjee2022,Flatt2022,Carceller2024,Flatt2025}.

\section{Parity-oblivious Hidden Matching}\label{sec:pohm}
\noindent
Throughout, $n$ is an even positive integer, and $[n]:={1,\dots,n}$. For $x\in\{0,1\}^n$ and $S\subseteq[n]$, we write
\begin{subequations}
\begin{align}
&x_S:=\bigoplus_{i\in S}x_i,\quad\text{with}\quad x_\emptyset:=0;\\
&\chi_S(x):=(-1)^{x_S}=(-1)^{\bigoplus_{k\in S}x_k}.
\end{align}
\end{subequations}
The function $\chi_S:\{0,1\}^n\to\{\pm1\}$ is called the Boolean character associated with $S$~\cite{deWolf2008}. A perfect matching of $[n]$ is a set $\mathtt{M}\subseteq\binom{[n]}{2}$ consisting of $n/2$ pairwise disjoint pairs whose union is $[n]$. We denote by $\mathscr{M}_n$ the collection of all perfect matchings of $[n]$. A simple counting argument gives $|\mathscr{M}_n|=(n-1)\times(n-3)\cdots3\times1=(n-1)!!$. We now recall the well-known Hidden Matching task introduced in~\cite{BarYossef2008}.

\begin{figure*}[t!]
\centering
\includegraphics[width=1\linewidth]{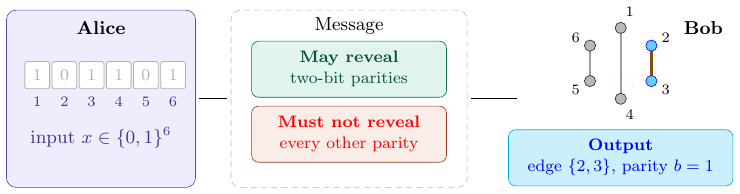}
\caption{(Color online) The task $\mathrm{PoHM}_n$, shown for $n=6$. Alice holds an $n$-bit string $x$ and sends Bob a single message, subject to the constraint that it reveal nothing about any parity $x_S$ with $|S|\neq0,2$; the two-bit parities $x_i\oplus x_j$ are exempt. Bob holds one of the $(n-1)!!$ perfect matchings of $[n]$, here $\mathtt{M}\equiv\{\{1,4\},\{2,3\},\{5,6\}\}\in\mathscr{M}_6$, and must return one of its edges together with the parity of the two bits it joins. Which edge he answers about is his choice, and Alice does not learn the matching.}\vspace{0cm}
\label{fig1}
\end{figure*}

\begin{definition}[Hidden Matching, $\mathrm{HM}_n$]\label{def:hm}
Alice receives $x\in\{0,1\}^n$, while Bob receives a perfect matching $\mathtt{M}\in\mathscr{M}_n$. Alice sends Bob a single classical message. Based on the message, shared randomness, and his input $\mathtt{M}$, Bob outputs a pair $\{i,j\}\in\mathtt{M}$ together with a bit $b\in\{0,1\}$. The output is correct if $b=x_i\oplus x_j$.
\end{definition}

\noindent The $\mathrm{HM}_n$ task is a one-way communication complexity problem~\cite{Yao1979,Kushilevitz1996,Buhrman2010}: communication is allowed from Alice to Bob, but not from Bob to Alice. The parties may additionally use local randomness and shared randomness. A classical protocol for $\mathrm{HM}_n$ is formally defined as follows.

\begin{definition}[Classical protocol]\label{def:protocol}
A classical protocol for $\mathrm{HM}_n$ consists of the following:
\begin{enumerate}[itemsep=-.05cm, topsep=2pt, leftmargin=.5cm]
\item a random variable $R$ representing the public randomness, independent of $X$ and taking values in a countable set;
\item a conditional distribution $\Pr(\mathtt{Msg}=\zeta\mid X=x,R=r)$, according to which Alice generates her message, where the message alphabet is arbitrary and, in particular, may be unbounded;
\item for each $(\zeta,r)$ and each $\mathtt{M}\in\mathscr{M}_n$, a probability distribution on pairs $(\{i,j\},b)$ with $\{i,j\}\in\mathtt{M}$, according to which Bob generates his output.
\end{enumerate}
Alice's input is a uniformly distributed random variable $X\sim\mathrm{Unif}\{0,1\}^n$, and is independent of Bob's input matching. We write $C:=(\mathtt{Msg},R)$ for the total information available to Bob about Alice's input through the protocol, and we use $X_S$ to denote the random variable $X_S:=\bigoplus_{i\in S}X_i$.
\end{definition}

\noindent 
The essential point in Definition~\ref{def:protocol} is that Bob's output distribution depends only on $(C,\mathtt{M})$ and not directly on Alice's input $x$. Bob's private randomness is absorbed into item~(3), so allowing such randomness entails no loss of generality. Likewise, there is no loss in taking $R$ to be public rather than shared but hidden, since $R$ is included in $C$ and is therefore available to Bob.

\begin{definition}[Zero error]\label{def:zeroerror}
A classical protocol solves $\mathrm{HM}_n$ with zero error if, for every $x\in\{0,1\}^n$, every $c$ satisfying $\Pr(C=c\mid X=x)>0$, every $\mathtt{M}\in\mathscr{M}_n$, and every output $(\{i,j\},b)$ that Bob produces with positive probability conditional on $(C,\mathtt{M})=(c,\mathtt{M})$, one has $b=x_i\oplus x_j$.
\end{definition}

\noindent It is worth noting that $\mathrm{HM}_n$ is a relational problem. Particularly, Bob may output any edge $\{i,j\}\in\mathtt{M}$ together with the corresponding parity $x_i\oplus x_j$; hence, for a fixed input $(x,\mathtt{M})$, there are generally several valid outputs. Bar-Yossef, Jayram, and Kerenidis showed that this relational problem admits a zero-error quantum one-way protocol using $O(\log n)$ qubits, whereas every classical randomized one-way protocol with bounded error requires $\Omega(\sqrt{n})$ bits, thereby establishing an exponential separation between quantum and classical one-way communication~\cite{BarYossef2008}. Since its introduction, Hidden Matching and its variants, including its Boolean decision variant, have been studied in several communication and nonlocality settings~\cite{Kerenidis2006,Gavinsky2007,Verbin2011,Buhrman2012,Kumar2019,Doriguello2020,Kapralov2022,Gilboa2026,Chakraborty2026}

Here we introduce another novel variant of the Hidden Matching task in which the communication is subject to a substantially stronger information constraint. Specifically, Alice's message is required to reveal no information about any parity of the input other than the two-bit parities.

\begin{definition}[Parity obliviousness, $\mathrm{PoHM}_n$]\label{def:po}
Let
\begin{align}
\mathcal{F}_n:=\bigl\{\,S\subseteq[n] \;:\; |S|\neq 0,2 \,\bigr\}.
\end{align}
A protocol is parity-oblivious if no measurement available to Bob yields any information about a non-exempt parity. Formally, writing $P_{S,b}$ for the procedure that samples $x$ uniformly from $\{x:x_S=b\}$ and implements the preparation procedure $P_x$ corresponding to input $x$
\begin{align}\label{eq:po}
p(k\mid P_{S,0},M)=p(k\mid P_{S,1},M)\quad \forall M,k,\ \forall S\in\mathcal{F}_n .
\end{align}
For a classical protocol, in which Bob's total information is the message $C$ of Definition~\ref{def:protocol}, condition~\eqref{eq:po} is equivalent to
\begin{align}\label{eq:poclassical}
I(X_S:C)=0 \qquad\text{for every } S\in\mathcal{F}_n .
\end{align}
Thus, the information available to Bob through $C$ is statistically independent of every nontrivial parity $X_S$ except those involving exactly two bits. The two-bit parities $X_{\{i,j\}}=X_i\oplus X_j$ are therefore exempted from the parity-obliviousness condition. The task Parity-oblivious Hidden Matching, denoted by $\mathrm{PoHM}_n$, is the $\mathrm{HM}_n$ task subject to the parity-obliviousness condition~\eqref{eq:po} [see Fig.~\ref{fig1}].
\end{definition}

\noindent 
The constraint in $\mathrm{PoHM}_n$ has the same rationale as in POM~\cite{Spekkens2009}, applied to a different query. There Bob is asked for a single bit $x_y$, so the single-bit parities $|S|=1$ are exempt and every parity of two or more bits is forbidden. Here Bob is asked for the parity of an edge, so it is the two-bit parities that are exempt and everything else that is forbidden. In particular, $\mathrm{PoHM}_n$ forbids information about individual bits---exactly what POM permits---since a message revealing $x_i$ and $x_j$ separately would yield $x_i\oplus x_j$ without being asked for it. In both cases the rule is the same: the message may carry what the task requires of it, and nothing more.

Since $X$ in Definition~\ref{def:po} is uniformly distributed, condition~\eqref{eq:po} is equivalent to the vanishing of the corresponding Fourier coefficients of the conditional distribution. More precisely, for every $c$ with $\Pr(C=c)>0$, $\widehat{p(c\mid\cdot)}(S)=0$ for every $S\in\mathcal{F}_n$, where for $f:\{0,1\}^n\to\mathbb{R}$, $\widehat{f}(S):=\tfrac{1}{2^n}\sum f(x)\chi_S(x)$. Equivalently, for each such $c$, the function
$x\mapsto\Pr(C=c\mid X=x)$ has a Fourier expansion containing only terms of degree at most $2$, with all linear terms absent. Thus it can be written in the form
\begin{align*}
\Pr(C=c\mid X=x)=a_\emptyset(c)+\sum_{\substack{S\subseteq[n]:|S|=2}}
a_S(c)\,\chi_S(x),
\end{align*}
with no terms corresponding to $|S|=1$ or $|S|>2$.

\subsection*{Three Useful Lemmas}

\noindent 
Fix a protocol and a value $c$ with $\Pr(C=c)>0$, and define
\begin{align}
\mathcal{A}_c:=\bigl\{x\in\{0,1\}^n:\Pr(C=c\mid X=x)>0\bigr\}.
\end{align}
Since $X$ is uniformly distributed and $\Pr(C=c)>0$, Bayes' rule gives
\begin{align}\label{eq:support}
\left.\begin{aligned}
\Pr(X=x\mid C=c)>0\quad\Longleftrightarrow\quad x\in\mathcal{A}_c,\\
\text{consequently, } \Pr\bigl(X\in\mathcal{A}_c\mid C=c\bigr)=1
\end{aligned}\right\}.
\end{align}

\noindent 
Define a relation $\sim_c$ on $[n]$ by
\begin{align}
i\sim_c j
\quad\Longleftrightarrow\quad
x_i\oplus x_j
\text{ is constant over }x\in\mathcal{A}_c.
\end{align}
The relation is reflexive since $x_i\oplus x_i=0$ for every $x$, and symmetric since
$x_i\oplus x_j=x_j\oplus x_i$. To see transitivity, suppose $i\sim_c j$ and $j\sim_c k$. Then both $x_i\oplus x_j$ and $x_j\oplus x_k$ are constant over $\mathcal{A}_c$, and hence $x_i\oplus x_k=(x_i\oplus x_j)\oplus(x_j\oplus x_k)$ is also constant over $\mathcal{A}_c$. Thus $\sim_c$ is an equivalence relation on $[n]$, and its equivalence classes will be called the $c$-classes.

\begin{lemma}[Zero error forces a fixed edge]\label{lem:fixededge}
Assume that the protocol solves $\mathrm{HM}_n$ with zero error. Then, for every $c$ with $\Pr(C=c)>0$ and every  $\mathtt{M}\in\mathscr{M}_n$, there exists an edge $\{i,j\}\in\mathtt{M}$ such that $i\sim_c j$.
\end{lemma}
\begin{proof}
Fix $c$ with $\Pr(C=c)>0$ and $\mathtt{M}\in\mathscr{M}_n$. Bob's output distribution conditional on $(C,\mathtt{M})=(c,\mathtt{M})$ is a probability distribution, and hence has nonempty support. Choose $(\{i,j\},b)$ in this support, where necessarily $\{i,j\}\in\mathtt{M}$. Crucially, this choice depends only on $(c,\mathtt{M})$ and not on $x$.

Now let $x\in\mathcal{A}_c$ be arbitrary. By the definition of $\mathcal{A}_c$, $\Pr(C=c\mid X=x)>0$. Since $(\{i,j\},b)$ occurs with positive probability conditional on $(C,\mathtt{M})=(c,\mathtt{M})$, the zero-error condition in Definition~\ref{def:zeroerror} implies $b=x_i\oplus x_j$. Because $b$ is fixed independently of $x$, the quantity $x_i\oplus x_j$ has the same value for every $x\in\mathcal{A}_c$. Hence $i\sim_c j$.
\end{proof}

\begin{lemma}[Obliviousness bounds the classes]\label{lem:classsize}
Assume the protocol satisfies $I(X_S:C)=0$ for every $S\subseteq[n]$ with $|S|=4$. Then, for every $c$ with $\Pr(C=c)>0$,
\begin{enumerate}[itemsep=-.04cm, topsep=0pt, leftmargin=.7cm]
\item[\emph{(i)}] every $c$-class has size at most $3$, and
\item[\emph{(ii)}] at most one $c$-class has size $\ge2$.
\end{enumerate}
Equivalently, the pairs whose parity is fixed by $c$ are exactly the
edges of a single clique on at most three vertices.
\end{lemma}

\begin{proof}
Both parts follow from the same observation: if $c$ fixes two disjoint pairwise parities, it fixes their XOR, which is a four-bit parity. Let $i,j,k,\ell\in[n]$ be distinct and suppose $x_i\oplus x_j$ and $x_k\oplus x_\ell$ are each constant over $\mathcal{A}_c$, say equal to $\alpha$ and $\beta$. XORing them gives
\begin{align}\label{eq:fourfixed}
x_i\oplus x_j\oplus x_k\oplus x_\ell=\alpha\oplus\beta=:y_c\quad\text{for all } x\in\mathcal{A}_c .
\end{align}
Putting $S:=\{i,j,k,\ell\}$, so that $|S|=4$, \eqref{eq:support} and
\eqref{eq:fourfixed} give
\begin{align}\label{eq:deterministic}
\Pr\bigl(X_S=y_c\mid C=c\bigr)=1 .
\end{align}
On the other hand, $X$ is uniformly distributed on $\{0,1\}^n$ and $S\neq\emptyset$, so $X_S$ is an unbiased bit; and $I(X_S:C)=0$ makes $X_S$ and $C$ independent. Hence
\begin{align}\label{eq:unbiased}
\Pr\bigl(X_S=0\mid C=c\bigr)=\Pr\bigl(X_S=1\mid C=c\bigr)=\tfrac12 ,
\end{align}
contradicting \eqref{eq:deterministic}. Thus no four distinct indices
can have two disjoint pairwise parities fixed.

\emph{(i)} Suppose some $c$-class $K$ satisfies $|K|\ge4$ and choose distinct $i,j,k,\ell\in K$. Since all four lie in one class, both $x_i\oplus x_j$ and $x_k\oplus x_\ell$ are constant over $\mathcal{A}_c$, which is excluded above. Hence $|K|\le3$.

\emph{(ii)} Suppose two distinct $c$-classes $K_1\neq K_2$ both had size at least $2$, and choose distinct $i,j\in K_1$ and distinct $k,\ell\in K_2$. The four indices are distinct, because $K_1$ and $K_2$ are disjoint, and again both $x_i\oplus x_j$ and $x_k\oplus x_\ell$ are constant over $\mathcal{A}_c$, which is excluded. Hence at most one class has size $\ge2$.

Finally, a pair $\{i,j\}$ has its parity fixed by $c$ precisely when $i\sim_c j$, so by (i) and (ii) the fixed pairs are the edges of the unique class of size $\ge2$, a clique on at most three vertices (and there are none if every class is a singleton).
\end{proof}

\begin{figure}[t!]
\centering
\includegraphics[width=1\linewidth]{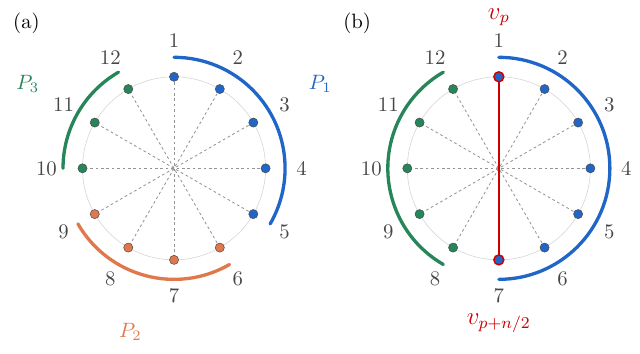}
\caption{(Color online) The antipodal matching, illustrated for $n=12$. (a) The elements of $[n]$ are placed around a circle so that each part of $\mathscr{P}\equiv P_1\sqcup P_2\sqcup P_3$ occupies a consecutive arc (here $|P_1|=5$, $|P_2|=4$, $|P_3|=3$, all at most $n/2$). The dashed diameters are the edges $\{v_p,v_{p+n/2}\}$ of $\mathtt{M}$; none has both endpoints in one part. (b) If a part spans more than $n/2$ positions, it contains an antipodal pair (red), since any arc through both $p$ and $p+n/2$ has at least $n/2+1$ positions.}\vspace{0cm}
\label{fig2}
\end{figure}

\begin{lemma}[Small parts admit an avoiding matching]\label{lem:matching}
Let $n$ be even and let $\mathscr{P}$ be a partition of $[n]$ whose parts all have size at most $n/2$. Then there exists $\mathtt{M}\in\mathscr{M}_n$ such that no edge of $\mathtt{M}$ has both endpoints in the same part of $\mathscr{P}$.
\end{lemma}

\begin{proof}
Arrange the elements of $[n]$ around a circle at positions $1,2,\dots,n$, ordering them so that the elements of each part of $\mathscr{P}$ occupy consecutive positions; this can be done by concatenating the parts in any order (see Fig.~\ref{fig2}). Let $v_p$ denote the element at position $p$, with indices understood modulo $n$, and define
\begin{align}
\mathtt{M}:=\Bigl\{\bigl\{v_p,v_{p+n/2}\bigr\}:1\le p\le n/2\Bigr\}.
\end{align}
Since each position $p$ is paired with the position antipodal to it on the circle, $\mathtt{M}$ is a perfect matching of $[n]$.

Suppose, for contradiction, that $v_p$ and $v_{p+n/2}$ belong to a common part $P\in\mathscr{P}$. The positions occupied by $P$ form a circular arc. Any circular arc containing both antipodal positions $p$ and $p+n/2$ must contain at least one of the two arcs joining them, namely $\{p,p+1,\dots,p+n/2\}$ or $\{p+n/2,p+n/2+1,\dots,p+n\}$, with indices taken modulo $n$. Each of these arcs contains $n/2+1$ positions (see Fig.~\ref{fig2}). Consequently, $|P|\ge n/2+1>n/2$, contradicting the assumption that every part of $\mathscr{P}$ has size at most $n/2$. Hence no edge of $\mathtt{M}$ has both endpoints in the same part of $\mathscr{P}$.
\end{proof}

\subsection{Classical Impossibility Theorem}

\begin{theorem}[Classical impossibility]\label{thm:cl}
Let $n\ge6$ be even. No classical protocol, with message alphabet of arbitrary size and with arbitrary public and private randomness, can simultaneously
\begin{enumerate}[itemsep=-.05cm, topsep=2pt, leftmargin=.5cm]
\item[$(1)$] solve $\mathrm{HM}_n$ with zero error, and
\item[$(2)$] satisfy $I(X_S:C)=0$ for every $S\subseteq[n]$ with $|S|=4$.
\end{enumerate}
In particular, the task $\mathrm{PoHM}_n$ admits no zero-error classical protocol for any even $n\ge6$.
\end{theorem}
\begin{proof}
Suppose, for contradiction, that a protocol satisfying $(1)$ and $(2)$ exists.

\noindent\emph{Fixing a message:} Since $C$ is a random variable, there exists some $c$ with $\Pr(C=c)>0$. Fix such a value of $c$. Let $\sim_c$ be the equivalence relation defined above, and let the corresponding equivalence classes be the $c$-classes.

\noindent\emph{The classes are small:} By hypothesis $(2)$ and Lemma~\ref{lem:classsize},
\begin{align}\label{eq:small}
|K|\le3\quad\text{for every $c$-class $K$.}
\end{align}
Indeed, if a $c$-class contained four distinct elements $i,j,k,\ell$, then the three pairwise parities $x_i\oplus x_j$, $x_i\oplus x_k$, and $x_i\oplus x_\ell$ would all be fixed on $\mathcal{A}_c$, and hence their XOR would fix the four-bit parity $x_i\oplus x_j\oplus x_k\oplus x_\ell$, contradicting the assumed four-bit parity obliviousness.

\noindent\emph{Constructing an avoiding matching:} Since $n\ge6$, one has $3\le\frac n2$. Thus, by \eqref{eq:small}, every $c$-class has size at most $n/2$. Lemma~\ref{lem:matching} therefore yields a perfect matching $\mathtt{M}\in\mathscr{M}_n$ such that no edge of $\mathtt{M}$ has both endpoints in the same $c$-class. Equivalently,
\begin{align}\label{eq:noedge}
\{i,j\}\in\mathtt{M}\quad\Longrightarrow\quad i\not\sim_c j.
\end{align}

\noindent\emph{The contradiction:} Now let Bob's input be this particular matching $\mathtt{M}$. Since the protocol solves $\mathrm{HM}_n$ with zero error, Lemma~\ref{lem:fixededge} implies that there must exist an edge $\{i,j\}\in\mathtt{M}$ such that $i\sim_c j$. This contradicts \eqref{eq:noedge}. Hence no protocol can satisfy both $(1)$ and $(2)$.

Finally, every subset $S\subseteq[n]$ with $|S|=4$ belongs to $\mathcal{F}_n$, since $|S|\neq0,2$. Therefore the parity-obliviousness condition in Definition~\ref{def:po} implies hypothesis $(2)$. Consequently, $\mathrm{PoHM}_n$ admits no zero-error classical protocol for any even $n\ge6$.
\end{proof}

\noindent 
A few crucial observations are in order: 
\begin{itemize}[itemsep=0cm, topsep=2pt, leftmargin=0cm]
\item[] \emph{Observation}-1: Firstly, the proof uses obliviousness only for subsets of size exactly $4$; no condition for $|S|=1$, $|S|=3$, or $|S|\ge5$ is invoked. Thus, Theorem~\ref{thm:cl} is strictly stronger than the impossibility statement for the parity-oblivious task $\mathrm{PoHM}_n$, whose definition imposes obliviousness for all $|S|\neq0,2$.

\item[] \emph{Observation}-2: Second, the threshold $n\ge6$ is sharp. For $n=4$, let Alice send $c=(x_1\oplus x_2,\;x_1\oplus x_3)$. Each of the three perfect matchings of $[4]$ contains at least one of the edges $\{1,2\}$, $\{1,3\}$, and $\{2,3\}$, so Bob can use $c$ to output a correct pair and its parity with certainty. Moreover, $c$ depends only on two-bit parities, and a direct check shows that $X_S$ remains unbiased conditioned on $c$ for every $S$ with $|S|\in\{1,3,4\}$. Hence this protocol is parity-oblivious. In accordance with Lemma~\ref{lem:classsize}, it has a unique nontrivial $c$-class, namely $\{1,2,3\}$, of size exactly $3$ --— the largest the lemma permits.
\end{itemize}

\begin{remark}[Necessity of even-parity condition]\label{rem:even}
The even-parity condition is essential to the argument. If obliviousness were imposed only for odd $|S|$, then, for uniform $X$, the condition would be equivalent to
\begin{align*}
\Pr(C=c\mid X=x)=\Pr(C=c\mid X=\bar x)\quad\text{for all }x,
\end{align*}
where $\bar x:=x\oplus1^n$. Consider the encoding $c=\bigl(x_1\oplus x_2,\ldots,x_1\oplus x_{n/2+1}\bigr)$. This encoding is invariant under $x\mapsto\bar x$ and therefore satisfies the odd-parity obliviousness condition. Let $A:=\{1,\ldots,n/2+1\}$. Every perfect matching $\mathtt{M}$ of $[n]$ contains an edge $\{i,j\}\in\mathtt{M}$ with $i,j\in A$, since $|A|=n/2+1>n/2$. From $c$, Bob can determine $x_i\oplus x_j$ for every $i,j\in A$: if $i=1$ or $j=1$, the parity is directly encoded, while for $i,j\neq1$, $x_i\oplus x_j=(x_1\oplus x_i)\oplus(x_1\oplus x_j)$. Hence Bob can always choose an edge of $\mathtt{M}$ contained in $A$ and output its parity with zero error. Thus odd-parity obliviousness alone does not preclude a zero-error classical protocol and cannot yield an impossibility result of the form in Theorem~\ref{thm:cl}.
\end{remark}

\begin{corollary}[$n/2$ bits are optimal]\label{coro:n/2optimal}
Let $n$ be even. The zero-error classical one-way communication complexity of the unconstrained task $\mathrm{HM}_n$, measured as the worst-case message length for each fixed value of the public randomness, is exactly $n/2$ bits.
\end{corollary}

\begin{proof}
The protocol described in Remark~\ref{rem:even} shows that $n/2$ bits suffice. We prove the converse.

Consider any zero-error classical protocol with public randomness $R$, and fix any value $r$ with $\Pr(R=r)>0$. Since $R$ is independent of $X$, conditioning on $R=r$ leaves $X$ uniformly distributed on $\{0,1\}^n$. For each message value $\zeta$ occurring with positive probability for at least one input when $R=r$, set $c:=(\zeta,r)$ and consider the corresponding set $\mathcal{A}_c$.

By Lemma~\ref{lem:fixededge} and Lemma~\ref{lem:matching}, the $c$-classes cannot all have size at most $n/2$. Hence there is a $c$-class $K_c$ satisfying $|K_c|\ge\frac n2+1$. Fix $i_0\in K_c$. Since $i_0\sim_c i$ for every $i\in K_c$, the value of $c$ determines each of the $|K_c|-1$ parities $X_{i_0}\oplus X_i$ for $i\in K_c\setminus\{i_0\}$. These parities are independent, since they are the parities associated with the edges of a spanning star on $K_c$. Therefore
\begin{align*}
|\mathcal{A}_c|\le2^{\,n-(|K_c|-1)}\le2^{\,n/2}.
\end{align*}
For the fixed value $r$, the sets $\mathcal{A}_{(\zeta,r)}$, over all possible message values $\zeta$, cover $\{0,1\}^n$: for every $x$, the conditional distribution $\Pr(\mathtt{Msg}=\zeta\mid X=x,R=r)$ is a probability distribution and hence has nonempty support. Thus
\begin{align*}
2^n\le\sum_{\zeta}|\mathcal{A}_{(\zeta,r)}|\le N_r\,2^{n/2},
\end{align*}
where $N_r$ is the number of possible message values for the fixed public-coin value $r$. Hence $N_r\ge 2^{n/2}$. In particular, if Alice's message uses at most $m$ bits, then $N_r\le2^m$, and therefore $m\ge\frac n2$. Combining this lower bound with the $n/2$-bit protocol from Remark~\ref{rem:even} proves the claim.
\end{proof}

\subsection{$\mathrm{PoHM}_n$ In Preparation Noncontextual Theories}
\noindent 
We first formalize how the operational parity-obliviousness constraint transfers to the ontological level. For each input $x$, let $P_x$ denote the corresponding preparation procedure. The parity-obliviousness condition implies operational equivalences between suitable mixtures of the preparations ${P_x}$. Preparation noncontextuality then requires the corresponding mixtures to have identical ontological representations.

\begin{lemma}[Obliviousness transfers to $\Lambda$]\label{lem:transfer}
Consider a preparation-noncontextual ontological model of a protocol for $\mathrm{HM}_n$ that is parity-oblivious in the sense of~\eqref{eq:po}, with $X$ uniformly distributed and $P_x$ denoting the preparation associated with input $x$. Let $\Lambda$ denote the corresponding ontic variable. Then $I(X_S:\Lambda)=0$ for every $S\in\mathcal{F}_n$.
\end{lemma}

\begin{proof}
Fix $S\in\mathcal{F}_n$ and $b\in\{0,1\}$, and let $H_{S,b}:=\{x\in\{0,1\}^n:x_S=b\}$. Since $S\neq\emptyset$, exactly half of the strings in $\{0,1\}^n$ satisfy $x_S=b$, and hence $|H_{S,b}|=2^{n-1}$. Let $P_{S,b}$ denote the preparation procedure that first samples $x$ uniformly from $H_{S,b}$ and then implements $P_x$. By Feature~\ref{fea:mixture}, its ontic distribution is
\begin{align}\label{eq:mixrep}
\mu(\lambda\mid P_{S,b})=\tfrac{1}{2^{n-1}}\sum_{x\in H_{S,b}}\mu(\lambda\mid P_x).
\end{align}
Likewise, for any measurement $M$ and outcome $k$, the operational statistics of the mixture are
\begin{align}\label{eq:mixstats}
p(k\mid P_{S,b},M)=\tfrac{1}{2^{n-1}}\sum_{x\in H_{S,b}}p(k\mid P_x,M).
\end{align}
The parity-obliviousness condition for $S$ means that, for every measurement $M$ and outcome $k$, the probability of obtaining $k$ is independent of the value of $X_S$. Since $X$ is uniformly distributed, this is precisely the statement
\begin{align}
p(k\mid P_{S,0},M)=p(k\mid P_{S,1},M)\quad\forall~M,k.
\end{align}
Thus, $P_{S,0}\simeq P_{S,1}$. By preparation noncontextuality, $\mu(\lambda\mid P_{S,0})=\mu(\lambda\mid P_{S,1})$ for almost every $\lambda$. Substituting \eqref{eq:mixrep} for $b=0$ and $b=1$ gives
\begin{align}
\tfrac{1}{2^{n-1}}\sum_{x:x_S=0}\mu(\lambda\mid P_x)=\tfrac{1}{2^{n-1}}\sum_{x:x_S=1}\mu(\lambda\mid P_x)
\end{align}
for almost every $\lambda$. We now identify these two sides with the conditional ontic distributions of $\Lambda$. Since $X$ is uniform, $\Pr(X=x\mid X_S=b)=2^{n-1}~\text{for every }x\in H_{S,b}$. Therefore, by the law of total probability at the ontological level, $\mu(\lambda\mid X_S=b)=\sum_{x\in H_{S,b}}\mu(\lambda\mid X=x)\Pr(X=x\mid X_S=b)=\frac{1}{2^{n-1}}\sum_{x\in H_{S,b}}\mu(\lambda\mid P_x)$. Consequently,
\begin{align}
\mu(\lambda\mid X_S=0)=\mu(\lambda\mid X_S=1)
\end{align}
for almost every $\lambda$. Because $X_S$ is an unbiased bit, the equality of the two conditional distributions implies that $\Lambda$ and $X_S$ are statistically independent: $\Pr(\Lambda\in A\mid X_S=0)=\Pr(\Lambda\in A\mid X_S=1)=\Pr(\Lambda\in A)$ for every measurable set $A\subseteq\Lambda$. Hence
\begin{equation}
I(X_S:\Lambda)=0.
\end{equation}
Since $S\in\mathcal{F}_n$ was arbitrary, the result holds for every $S\in\mathcal{F}_n$.
\end{proof}

\noindent
Two features of the argument are worth emphasizing. First, parity obliviousness is formulated entirely in terms of operationally accessible statistics and is therefore, in principle, directly testable without reference to any underlying ontological description. In this respect, it is analogous to the no-signaling principle in Bell scenarios, where the relevant constraint is likewise imposed at the operational level. Second, the operational equivalences $P_{S,0}\simeq P_{S,1}$ are nontrivial: they relate distinct preparation procedures that yield identical statistics for every allowed measurement. It is precisely such operational indistinguishability between distinct procedures that gives preparation noncontextuality substantive force, requiring their ontological representations to coincide and thereby imposing constraints on the underlying ontological model.

\begin{theorem}\label{thm:pnc}
For every even $n\ge 6$, no preparation-noncontextual ontological model can win the $\mathrm{PoHM}_n$ task with certainty.
\end{theorem}

\begin{proof}
Suppose, for contradiction, that an operational theory admits a preparation-noncontextual ontological model $(\Lambda,\mu,\xi)$ and contains a protocol that is parity-oblivious and solves $\mathrm{PoHM}_n$ with certainty. For each $x\in\{0,1\}^n$, let $P_x$ denote the corresponding preparation. For each perfect matching $\mathtt{M}\in\mathscr{M}_n$, let the measurement performed by Bob be denoted by $\mathtt{M}$, with outcomes identified with the corresponding pair-and-bit $(\{i,j\},b)$. We use the ontic state itself as the message of an induced classical protocol and derive a contradiction with Theorem~\ref{thm:cl}.

\smallskip
\noindent\emph{Step 1 (The induced classical protocol $\Pi_\Lambda$) --} On input $x$, Alice samples $\lambda\in\Lambda$ according to $\mu(\lambda\mid P_x)$ and sends $\lambda$ to Bob. Thus, the classical message is $C:=\Lambda$. Upon receiving $\lambda$ and $\mathtt{M}$, Bob samples an outcome $k$ according to the response function $\xi(k\mid \lambda,\mathtt{M})$ and outputs the pair-and-bit associated with $k$. This defines a classical one-way protocol: Bob's output distribution depends only on the received message $\lambda$ and his input $\mathtt{M}$, and not on $x$. The message space may be arbitrarily large, which is allowed in Theorem~\ref{thm:cl}.

\smallskip
\noindent\emph{Step 2 ($\Pi_\Lambda$ solves $\mathrm{HM}_n$ with zero error)--} Fix $x\in\{0,1\}^n$ and $\mathtt{M}\in\mathscr{M}_n$. Let $(\{i,j\},b)$ be an incorrect output for $(x,\mathtt{M})$, so that $b\neq x_i\oplus x_j$. Since the original protocol solves $\mathrm{PoHM}_n$ with certainty, the probability of this incorrect outcome is zero: $p\bigl((\{i,j\},b)\mid P_x,\mathtt{M}\bigr)=0$. By the ontological representation~\eqref{eq:ontrep},
\begin{align*}
0&=\int_\Lambda\xi\bigl(\lambda\mid(\{i,j\},b)\,\mathtt{M}\bigr)\,\mu(\lambda\mid P_x)\,d\lambda .
\end{align*}
The integrand is nonnegative. Hence, for every $\lambda$ with $\mu(\lambda\mid P_x)>0$, we must have $\xi\bigl(\lambda\mid(\{i,j\},b),\mathtt{M}\bigr)=0$. Therefore, whenever Alice's input is $x$ and the induced classical protocol receives a message $\lambda$ having positive probability, Bob's response distribution assigns zero probability to every incorrect output. Equivalently, every output produced with positive probability is correct: $b=x_i\oplus x_j$. Thus $\Pi_\Lambda$ satisfies the zero-error condition of Definition~\ref{def:zeroerror}.

\smallskip
\noindent\emph{Step 3 ($\Pi_\Lambda$ is parity-oblivious).}
By assumption, the original protocol is operationally parity-oblivious, and its ontological model is preparation noncontextual. Lemma~\ref{lem:transfer} therefore gives $I(X_S:\Lambda)=0\quad\text{for every }S\in\mathcal{F}_n$. Since the message of $\Pi_\Lambda$ is precisely $C=\Lambda$, we have $I(X_S:C)=0$ for every $S\in\mathcal{F}_n$. In particular, $I(X_S:C)=0$ for every $S\subseteq[n]$ with $|S|=4$. Thus $\Pi_\Lambda$ satisfies the obliviousness hypothesis of Theorem~\ref{thm:cl}.

It is worth emphasizing the role of preparation noncontextuality. Operational parity obliviousness constrains only the statistics of the permitted measurements and, by itself, does not imply that the underlying ontic state is independent of the forbidden parities. Preparation noncontextuality bridges this gap: the operational equivalences $P_{S,0}\simeq P_{S,1}$ force equality of the corresponding ontic distributions, which is precisely what yields $I(X_S:\Lambda)=0$.

\smallskip
\noindent\emph{Step 4 (Contradiction)--} Steps 2 and 3 show that $\Pi_\Lambda$ is a classical protocol that simultaneously solves $\mathrm{HM}_n$ with zero error and satisfies $I(X_S:C)=0$ for every $|S|=4$. For even $n\ge6$, this contradicts Theorem~\ref{thm:cl}. Hence no preparation-noncontextual ontological model can solve $\mathrm{PoHM}_n$ with certainty.
\end{proof}

\begin{remark}[Bob's inefficiency is not needed]\label{rem:omniscient}
In Step~1, Bob is given the ontic state $\lambda$ itself, rather than accessing it only through the response functions $\xi(k\mid \lambda,\mathtt{M})$. Thus, even an \emph{omniscient decoder} with unrestricted access to $\lambda$ cannot solve $\mathrm{HM}_n$ under the parity-obliviousness constraint. The obstruction therefore lies in the information content of $\lambda$ itself, not in any limitation of the decoding procedure.
\end{remark}

\begin{remark}[Continuous ontic state spaces]\label{rem:measurability}
Although Theorem~\ref{thm:cl} is stated for countable message alphabets, the argument extends to arbitrary ontic spaces by coarse-graining. Since $X$ takes values in the finite set $\{0,1\}^n$, define, for almost every $\lambda$,
\begin{align*}
\mathcal{A}_\lambda:=\bigl\{x\in\{0,1\}^n:\Pr(X=x\mid\Lambda=\lambda)>0\bigr\},
\end{align*}
and let $C':=\mathcal{A}_\Lambda$. Then $C'$ takes only finitely many values, being a subset of the finite set $\{0,1\}^n$. Moreover, since $C'$ is a function of $\Lambda$, $I(X_S:C')\le I(X_S:\Lambda)=0~\forall\,S\in\mathcal{F}_n$. For every $c'$ with $\Pr(C'=c')>0$, the set of inputs having positive conditional probability given $C'=c'$ is exactly $c'$: $\Pr(X=x\mid C'=c')>0\Longleftrightarrow x\in c'$. So the equivalence classes used in Theorem~\ref{thm:cl} are unchanged. Averaging Bob's response function over the ontic states compatible with $c'$ yields a valid finite-message classical protocol, and zero error is preserved. Thus Theorem~\ref{thm:cl} applies also when $\Lambda$ is continuous.
\end{remark}

\subsection{Quantum Feasibility}

\noindent
Theorems~\ref{thm:cl} and~\ref{thm:pnc} show that no preparation-noncontextual strategy can solve \(\mathrm{PoHM}_n\) perfectly, but this alone does not rule out the possibility that the task itself is impossible. The following theorem gives a quantum protocol that satisfies the parity-obliviousness constraint and wins \(\mathrm{PoHM}_n\) with certainty. Thus, the separation arises from preparation noncontextuality, rather than from an intrinsic limitation of the task. 

\begin{theorem}\label{thm:qm}
For every even $n$ there is a quantum protocol using $\lceil\log_2 n\rceil$ qubits that solves $\mathrm{HM}_n$ with zero error and is parity-oblivious in the sense that, for every $S\in\mathcal{F}_n$ and every POVM $\{E_k\}_k$ performed by Bob, the outcome statistics are independent of $x_S$.
\end{theorem}

\begin{proof}
Let $q:=\lceil\log_2 n\rceil$, so that $2^q\ge n$. Depending upon her string $x\in\{0,1\}^n$, Alice prepares the $n$-dimensional state
\begin{align}\label{eqq1}
x\mapsto|\psi_x\rangle=\tfrac{1}{\sqrt n}\sum_{i=1}^n(-1)^{x_i}|i\rangle\in\mathbb{C}^n,
\end{align}
and sends it to Bob, encoded in a $q$-qubit register. The corresponding density operator is
\begin{align}\label{eqq2}
\rho_x&=\tfrac1n\sum_{i,j}(-1)^{x_i+x_j}|i\rangle\langle j|\nonumber\\
&=\tfrac1n \mathbf{I}+\tfrac1n\sum_{\{i,j\}\in\binom{[n]}2}
\chi_{\{i,j\}}(x)\bigl(|i\rangle\langle j|+|j\rangle\langle i|\bigr),
\end{align}
where the diagonal terms $i=j$ contribute $(-1)^{2x_i}=1$ and have been collected into $\mathbf{I}/n$. Equation~\eqref{eqq2} may be viewed as a Fourier expansion of the operator-valued function $x\mapsto\rho_x$: the coefficient of $\chi_\emptyset$ is $\mathbf{I}/n$, the coefficient associated with $\chi_{\{i,j\}}$ is $\frac1n(|i\rangle\langle j|+|j\rangle\langle i|)$, and every other Fourier coefficient vanishes. Thus, the dependence of $\rho_x$ on $x$ is entirely through characters indexed by two-element subsets.

\noindent\emph{Obliviousness--} Fix $S\in\mathcal{F}_n$ and $b\in\{0,1\}$, and let
\begin{align}
H_{S,b}:=\big\{x\in\{0,1\}^n: x_S=b\big\}.
\end{align}
Since $X$ is uniformly distributed on $\{0,1\}^n$, the conditional distribution of $X$ given $X_S=b$ is uniform on $H_{S,b}$. In particular, $|H_{S,b}|=2^{n-1}$. The set $H_{S,0}$ is an index-two subgroup of $\{0,1\}^n$ under bit-wise addition modulo $2$, while $H_{S,1}$ is its nontrivial coset. Since a nontrivial character sums to zero over a finite group, $\sum_{x\in H_{S,0}}\chi_\kappa(x)$ vanishes unless $\chi_\kappa$ restricts trivially to $H_{S,0}$, i.e.\ unless $\kappa$ lies in the annihilator $H_{S,0}^{\perp}=\{\varnothing,S\}$; the coset $H_{S,1}$ follows by translation, contributing the factor $\chi_\kappa(x^{(S)})=(-1)^b$ when $\kappa=S$. Hence
\begin{align}
\tfrac{1}{|H_{S,b}|} \sum_{x\in H_{S,b}}\chi_\kappa(x) = \begin{cases} 1, & \kappa=\varnothing,\\
(-1)^b, & \kappa=S,\\
0, & \kappa\neq\varnothing,S.
\end{cases}
\end{align}
Consequently, $\mathbb{E}\left[ \chi_\kappa(X)\mid X_S=b \right] = 0$, whenever $\kappa\neq\varnothing,S$. By linearity of expectation, 
\begin{align}\label{eqq3}
\rho_{S,b} &=\mathbb{E}\!\left[ \rho_X\mid X_S=b \right]\nonumber\\
&= \tfrac{1}{n}\mathbf{I} + \tfrac{1}{n} \sum_{\{i,j\}} \mathbb{E}\!\left[ \chi_{\{i,j\}}(X)\mid X_S=b \right]\bigl(\ket{i}\bra{j}+\ket{j}\bra{i}\bigr)\nonumber\\
&=\tfrac{1}{n}\mathbf{I}\quad\text{for every}\quad S\in\mathcal F_n.
\end{align}
Thus, $\operatorname{Tr}(\rho_{S,0}E_k)=\operatorname{Tr}(\rho_{S,1}E_k)$ for every $S\in\mathcal F_n$ and for every element $E_k$ of every POVM $\left\{E_k\right\}_k$. Hence the outcome statistics of any measurement performed by Bob are independent of $x_S$, as required. 

It is worth noting where the argument stops. For $|S|=2$, say $S=\{i,j\}$, the term $\chi_{\{i,j\}}$ survives the averaging and Eq.~\eqref{eqq3} generally fails. Thus the protocol does reveal the two-bit parities. This is not a defect, but precisely the intended behavior, since these include the parities Bob is required to report and $\mathcal{F}_n$ explicitly exempts them.

\noindent\emph{Perfect success--} The protocol is exactly the one proposed in~\cite{BarYossef2008}. Bob, after receiving the encoded quantum system from Alice together with the matching $\mathtt{M}$, performs a projective measurement in the basis
\begin{align}\label{eqq4}
\mathcal{B}_{\mathtt{M}}:=\Bigl\{\,|e^\pm_{ij}\rangle:=\tfrac{1}{\sqrt2}\bigl(|i\rangle\pm|j\rangle\bigr):\{i,j\}\in\mathtt{M}\,\Bigr\}.
\end{align}
Because $\mathtt{M}$ is a perfect matching, the vectors in $\mathcal{B}_{\mathtt{M}}$ form an orthonormal basis of $\mathbb{C}^n$. If $2^q>n$, complete $\mathcal{B}_{\mathtt{M}}$ arbitrarily to an orthonormal basis of the full $q$-qubit space; since $\ket{\psi_x}\in\mathbb{C}^n$, the added outcomes occur with probability zero. Taking inner products with the encoded state \eqref{eqq1},
\begin{align}
\langle e^{\pm}_{ij}|\psi_x\rangle&=\tfrac{1}{\sqrt{2n}}\Bigl((-1)^{x_i}\pm(-1)^{x_j}\Bigr).
\end{align}
Thus $|e^+_{ij}\rangle$ has nonzero amplitude only when $x_i\oplus x_j=0$, whereas $|e^-_{ij}\rangle$ has nonzero amplitude only when $x_i\oplus x_j=1$. Exactly one of the two vectors associated with each edge can therefore occur, and the sign of the observed vector determines $x_i\oplus x_j$ with certainty. Bob outputs the corresponding edge $\{i,j\}$ together with the inferred parity. Hence the protocol succeeds with probability $1$ for every $x\in\{0,1\}^n$ and every $\mathtt{M}\in\mathscr{M}_n$.
\end{proof}

\noindent 
Note that, neither of the two requirements—the prescribed information constraint on Alice's message and successful recovery of the relevant matching parity—is difficult to satisfy on its own. The empty message satisfies parity obliviousness, while $\mathrm{HM}_n$ can be solved with zero error using $n/2$ classical bits. What Theorem~\ref{thm:cl} and ~\ref{thm:qm} together show is that the conjunction of the two requirements separates the classical and quantum theories: for even $n\ge6$, it is achievable quantum-mechanically, whereas classically it is impossible at any communication cost. Thus the separation is qualitative rather than merely quantitative.

It is instructive to contrast $\mathrm{PoHM}_n$ with the parity-oblivious multiplexing (POM) task of Spekkens \emph{et al}~\cite{Spekkens2009}. In POM, preparation contextuality manifests itself through a quantitative noncontextuality inequality: quantum theory achieves a success probability strictly above the preparation-noncontextual bound, but the quantum success probability is itself strictly less than unity. In this sense, the POM advantage is analogous in spirit to inequality-based tests such as the CHSH test of nonlocality~\cite{Clauser1969}: the local/noncontextual model imposes a quantitative bound on the success probability, which is exceeded by quantum theory. By contrast, $\mathrm{PoHM}_n$ exhibits an \emph{all-vs-nothing}, or GHZ-type, form of contradiction~\cite{Greenberger1989,Mermin1990,Greenberger1990}. The quantum protocol succeeds with certainty, whereas every preparation-noncontextual model fails. Thus, the contradiction with preparation noncontextuality is already visible at the level of possible versus impossible events, rather than requiring a quantitative statistical excess over a classical threshold.

\section{A noncontextuality inequality}\label{sec:inequality}

\noindent 
So far we have established an \emph{all-vs-nothing} result for preparation contextuality in quantum theory: while Theorem~\ref{thm:qm} shows that $\mathrm{PoHM}_n$ can be perfectly won in quantum theory, Theorem~\ref{thm:pnc} shows that no preparation-noncontextual model can have perfect success whenever $n\ge6$. Importantly, the later theorem does not determine how close a preparation-noncontextual model can come to perfect success; it rules out only the extremal value $1$. In particular, it does not provide a uniform upper bound $\beta<1$ on the success probability.

A quantitatively stronger statement is a \emph{noncontextuality inequality} of the form
\begin{align}
p(\text{success})\le \beta_{\mathrm{NC}},
\end{align}
where $\beta_{\mathrm{NC}}<1$ holds for every preparation-noncontextual model satisfying the prescribed operational constraints. Such an inequality establishes a finite separation between the quantum value and the noncontextual bound and therefore provides a quantitative witness of preparation contextuality. In the present setting, the relevant operational constraint is parity obliviousness, which must itself be independently verified when interpreting the inequality experimentally. The purpose of this section is to determine the optimal noncontextual bound $\beta_{\mathrm{NC}}$ for the $\mathrm{PoHM}_n$ task.

\subsection{The Structure of Noncontextual Strategies}

\noindent 
We first characterize what a preparation-noncontextual theory can achieve in the $\mathrm{PoHM}_n$ task. Although the ontic state space $\Lambda$ may be arbitrarily large and Bob's measurements are unrestricted, the optimization can be reduced to a finite-dimensional problem. The key point is that parity obliviousness, once transferred to the ontological level by Lemma~\ref{lem:transfer}, strongly constrains how the ontic distribution can depend on the input. In particular, the conditional distribution of the ontic state, viewed as a function of the input $x$, can contain only the constant term and two-bit parity terms. Thus, from the perspective of the task, an arbitrary ontic state space offers no more power than a classical message whose conditional distribution is a nonnegative Fourier polynomial of degree at most $2$ with no linear terms. The remaining freedom is captured by the pairwise correlations $\langle\varepsilon_i\varepsilon_j\rangle$, subject to their arising from a valid probability distribution on $\{\pm1\}^n$. This reduces the problem to a finite optimization.

Throughout this section, we write $\varepsilon_i:=(-1)^{x_i}$ and identify functions on $\{0,1\}^n$ with functions on $\{\pm1\}^n$. For $g:\{\pm1\}^n\to\mathbb{R}$, we use the Fourier convention~\cite{deWolf2008}
\begin{align}\label{eq:fourier}
\widehat g(S):=\mathbb{E}_{\varepsilon}\Bigl[g(\varepsilon)\prod_{i\in S}\varepsilon_i\Bigr]
=2^{-n}\sum_{x}g(x)\chi_S(x),
\end{align}
so that $\widehat g(\emptyset)=\mathbb{E}g$. Write $\mathcal{E}_n:=\binom{[n]}{2}$ for the set of pairs and $\mathscr{M}_n$ for the set of perfect matchings of $[n]$. Throughout, $\mathtt{M}$ is sampled uniformly from $\mathscr{M}_n$.

\begin{definition}[Admissible densities]\label{def:admissible}
Let $\mathcal{K}_n$ denote the set of functions
$g:\{\pm1\}^n\to\mathbb{R}$ satisfying
\begin{align}\label{eq:Kn}
\left\{\begin{aligned}\emph{(a) } g(\varepsilon)\ge0 \ \ \forall\,\varepsilon\in\{\pm1\}^n,\quad\emph{(b) } \widehat{g}(\emptyset)=1,\\
\emph{(c) } \widehat{g}(S)=0 \ \ \forall\, S\in\mathcal{F}_n\qquad\qquad
\end{aligned}\right\}
\end{align}
Conditions \emph{(a)} and \emph{(b)} say exactly that $2^{-n}g$ is a probability distribution on $\{\pm1\}^n$, and \emph{(c)} that it is parity-oblivious; we call the elements of $\mathcal{K}_n$ \emph{admissible densities}.
\end{definition}

By (c), the only Fourier coefficients of $g\in\mathcal{K}_n$ that may be nonzero are those indexed by $\emptyset$ and by two-element subsets, so
\begin{align}\label{eq:Knexpansion}
g(\varepsilon)=1+\sum_{e\in\mathcal{E}_n}\widehat{g}(e)\,\varepsilon_e .
\end{align}
An admissible density is therefore determined by the vector $\bigl(\widehat{g}(e)\bigr)_{e\in\mathcal{E}_n}\in\mathbb{R}^{\binom n2}$, and we identify $\mathcal{K}_n$ with the set of such vectors. Under this identification (b) and (c) hold automatically, while (a) imposes one linear inequality for each $\varepsilon$; since $g(\varepsilon)=g(-\varepsilon)$, only $2^{n-1}$ of these are distinct. Thus $\mathcal{K}_n$ is a polytope with $2^{n-1}$ facets, and in particular is compact and convex. Finally, for every $e\in\mathcal{E}_n$,
\begin{align}\label{eq:coeffbound}
\bigl|\widehat{g}(e)\bigr|=\bigl|\mathbb{E}[g\,\varepsilon_e]\bigr|\le\mathbb{E}[g]=1,
\end{align}
using $g\ge0$ and $|\varepsilon_e|=1$.

\begin{definition}[Matching functional and noncontextual value]\label{def:Phi}
For $g\in\mathcal{K}_n$ define the \emph{matching functional}
\begin{align}\label{eq:Phi}
\Phi(g):=\mathbb{E}_{\mathtt{M}}
\Bigl[\max_{e\in\mathtt{M}}\bigl|\widehat{g}(e)\bigr|\Bigr],
\end{align}
the expectation being over a uniformly random
$\mathtt{M}\in\mathscr{M}_n$, and set
\begin{align}\label{eq:Rn}
R_n:=\max_{g\in\mathcal{K}_n}\Phi(g).
\end{align}
\end{definition}
\noindent
The maximum in \eqref{eq:Rn} is attained: $\Phi$ is a finite average of maxima of finitely many continuous functions of the coefficients, hence continuous, and $\mathcal{K}_n$ is compact. Note also that $\Phi$ is convex and positively homogeneous in $\widehat{g}$, properties used in
the proof of Proposition~\ref{prop:reduction}.

\begin{proposition}[Exact reduction]\label{prop:reduction}
Let $p^{\mathrm{NC}}_{\mathrm{opt}}(n)$ denote the supremum, over all operational theories admitting a preparation noncontextual hidden variable model and all parity-oblivious protocols therein, of the probability that Bob outputs a correct edge and parity in $\mathrm{PoHM}_n$. Then
\begin{align}\label{eq:reduction}
p^{\mathrm{NC}}_{\mathrm{opt}}(n)=\tfrac12+\tfrac12R_n .
\end{align}
\end{proposition}

\begin{proof}
\emph{Upper bound.} Fix a preparation-noncontextual model and a parity-oblivious protocol, and let $\Lambda$ denote its ontic variable. We write the argument for discrete $\Lambda$; for a general measurable $\Lambda$ the proof is unchanged upon replacing $\sum_\lambda$ by $\int d\lambda$ throughout, since $\Phi$ is applied to the normalized conditional density $g_\lambda/a_\lambda$. For each $\lambda$, define 
\begin{align*}
g_\lambda(x):=p(\lambda\mid P_x),\quad~a_\lambda:=\widehat{g_\lambda}(\emptyset)
\end{align*}
Since $X$ is uniform, $a_\lambda=2^{-n}\sum_x g_\lambda(x)=\Pr(\Lambda=\lambda)$, and, for every $x$, $\sum_\lambda g_\lambda(x)=1$. By Lemma~\ref{lem:transfer}, parity obliviousness together with preparation noncontextuality implies $I(X_S:\Lambda)=0$ for every $S\in\mathcal{F}_n$. Hence, for every $\lambda$ with $a_\lambda>0$,
\begin{align*}
\widehat{g_\lambda}(S)=0~\text{for all }S\in\mathcal{F}_n,~\text{and therefore~}\frac{g_\lambda}{a_\lambda}\in\mathcal{K}_n.
\end{align*}
For an edge $e=\{i,j\}$, the posterior distribution of $X$ conditioned on $\Lambda=\lambda$ is $p(x\mid\Lambda=\lambda)=\frac{2^{-n}g_\lambda(x)}{a_\lambda}$. Thus
\begin{align}\label{eq:condcorr}
\mathbb{E}\bigl[\varepsilon_i\varepsilon_j\mid\Lambda=\lambda\bigr]=\tfrac{\widehat{g_\lambda}(e)}{a_\lambda}.
\end{align}
Consequently, even an observer who is given $\lambda$ itself and is otherwise unrestricted can determine the parity on a fixed edge $e$ with success probability at most
\begin{align*}
\tfrac12\left(1+\tfrac{|\widehat{g_\lambda}(e)|}{a_\lambda}\right),
\end{align*}
with equality obtained by guessing the more likely parity. Since Bob in the original protocol has no more information than $\lambda$ together with $\mathtt{M}$, his success probability is bounded by that of such an observer who, after seeing $\lambda$ and $\mathtt{M}$, chooses the edge maximizing the corresponding bias. Therefore
\begin{align*}
p_{\mathrm{succ}}\le
\sum_\lambda a_\lambda\,
\mathbb{E}_{\mathtt{M}}\Big[\tfrac12\big(1+\max_{e\in \mathtt{M}}\tfrac{|\widehat{g_\lambda}(e)|}{a_\lambda}\big)\Big]=\tfrac12\big[1+\sum_\lambda
\Phi(g_\lambda)\big],
\end{align*}
where we used $\sum_\lambda a_\lambda=1$ and the positive homogeneity of $\Phi$. Since $\Phi(g_\lambda)=a_\lambda\,\Phi\!\big(\frac{g_\lambda}{a_\lambda}\big)\le a_\lambda R_n$, we obtain $p_{\mathrm{succ}}\le\tfrac12+\tfrac12R_n$.

\noindent
\emph{Attainability.} Let $g^\ast\in\mathcal{K}_n$ attain $R_n$. For each $\eta\in\{0,1\}^n$, define $g_\eta(x):=2^{-n}g^\ast(x\oplus\eta)$. Since $g^\ast\ge0$, we have $g_\eta\ge0$, and
\begin{align}
\sum_{\eta}g_\eta(x)&=2^{-n}\sum_{\eta}g^\ast(x\oplus\eta)\nonumber\\
&=2^{-n}\sum_y g^\ast(y)= \widehat{g^\ast}(\emptyset)=1.
\end{align}
Hence $\{g_\eta(x)\}$ defines a valid classical encoding: on input $x$, Alice sends the message $\eta$ with probability $g_\eta(x)$. Since $X$ is uniform, $\Pr(\eta)=2^{-n}\sum_x g_\eta(x)=2^{-n}$. Thus the messages are uniformly distributed marginally. Changing variables $y=x\oplus\eta$ in \eqref{eq:fourier} gives
\begin{align}\label{eq:shift}
\widehat{g_\eta}(S)=2^{-n}\chi_S(\eta)\,\widehat{g^\ast}(S).
\end{align}
Hence $\widehat{g_\eta}(S)=0$ for every $S\in\mathcal{F}_n$, so the encoding is parity-oblivious. Moreover, $|\widehat{g_\eta}(e)|=2^{-n}|\widehat{g^\ast}(e)|$ for every edge $e$. Upon receiving $\eta$ and $\mathtt{M}$, Bob chooses an edge $e\in\mathtt{M}$ maximizing $|\widehat{g_\eta}(e)|$ and outputs the more likely value of the corresponding parity, with the choice of parity determined by the sign of $\widehat{g_\eta}(e)$. His conditional success probability is therefore
\begin{align*}
\tfrac12\left(1+\max_{e\in\mathtt{M}}\tfrac{|\widehat{g_\eta}(e)|}{\widehat{g_\eta}(\emptyset)}\right).
\end{align*}
Since $\widehat{g_\eta}(\emptyset)=2^{-n}$, averaging over the uniformly distributed messages $\eta$ and over $\mathtt{M}$ gives
\begin{align}
p_{\mathrm{succ}}=\tfrac12+\tfrac12\sum_\eta\Phi(g_\eta)=\tfrac12+\tfrac12R_n.
\end{align}
Thus the upper bound is attainable, proving \eqref{eq:reduction}.
\end{proof}
\noindent
Proposition~\ref{prop:reduction} says that a preparation noncontextual model can do no better than a classical randomized encoding whose conditional density is a nonnegative multi-linear polynomial of degree at most two. The unbounded ontic state buys nothing.

\subsection{The Best Noncontextual Strategy}\label{ssec:beststrategy}

We now ask what an optimal noncontextual strategy should look like. The answer is somewhat different from what the zero-error analysis suggests. A natural strategy is for the ontic state to determine the parities among a small set of coordinates. Lemma~\ref{lem:classsize} shows that, under parity obliviousness, this set can contain at most three coordinates and Lemma~\ref{lem:classsize} also shows there can be only one such set, so the ontic state can fix the parities of a single triple and no more. Bob then succeeds with certainty whenever his matching contains an edge entirely within this set, and guesses otherwise. Although this strategy is natural from the zero-error perspective, it is not optimal for average success probability.

A better strategy trades certainty for coverage. Instead of determining a few parities exactly, the ontic state can encode weaker correlations spread over a larger set of coordinates\footnote{A similar trade-off is familiar from random access codes (RACs). In the classical $3\mapsto1$ RAC, encoding a single bit $x_1$ lets Bob answer $y=1$ with certainty and forces him to guess otherwise, for an average success of $1/3+2/3\times1/2=2/3$. Encoding instead the majority $\mathrm{maj}(x_1,x_2,x_3)$ determines no individual bit perfectly, yet answers a uniformly chosen $x_y$ correctly with probability $3/4$~\cite{Ambainis2024}. See Appendix \ref{app:n8} for more details.}. Bob can then obtain a nonzero bias for more potential edges, even though no individual parity is determined with certainty (see Fig.~\ref{fig3}). As we show below, the optimal trade-off balances the strength of the correlation, of order $2/(k-1)$, against the probability that the matching contains a correlated edge, leading to the scale $k\sim\sqrt n$. Thus, unlike the zero-error setting, where only perfect correlations matter, controlled partial correlations can improve the average success probability.

\begin{lemma}[Odd cliques are admissible]\label{lem:clique}
Let $k\ge3$ be odd and $K\subseteq[n]$ with $|K|=k$. Put $S_K:=\sum_{i\in K}\varepsilon_i$ and
\begin{align}\label{eq:clique}
g_K:=\tfrac{S_K^2-1}{k-1}.
\end{align}
Then $g_K\in\mathcal{K}_n$, and
\begin{align}\label{eq:cliquecoeff}
\widehat{g}_K(e)=\begin{cases}
\tfrac{2}{k-1}, & e\subseteq K,\\
0, & \text{otherwise}.
\end{cases}
\end{align}
\end{lemma}

\begin{figure}[t!]
\centering
\includegraphics[width=1\linewidth]{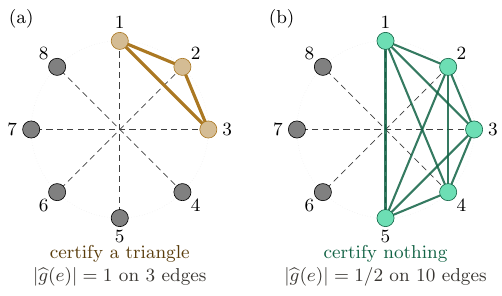}
\caption{(Color online) Certainty versus coverage, drawn for $n=8$. (a)~The triangle is the largest set whose pairwise parities the ontic state may fix outright, by Lemma~\ref{lem:classsize}; but only three of the $\binom82=28$ edges carry any correlation, and a uniform perfect matching meets one of them with probability only $3/7$. The dashed matching shown misses the triangle entirely, and Bob must then guess. (b)~The quintuple fixes no parity at all: every one of its ten edges carries the weaker bias $|\widehat g(e)|=1/2$. Since only three vertices lie outside, however, \emph{every} perfect matching has two of its endpoints inside the quintuple and hence contains one of those ten edges. Trading certainty for coverage raises the value from $3/7$ to $1/2$. The crossover occurs exactly at $n=8$.}\vspace{0cm}
\label{fig3}
\end{figure}

\begin{proof}
Expanding $S_K^2$ gives
\begin{align}\label{eq:cliqueexp}
g_K=1+\tfrac{2}{k-1}\sum_{\{i,j\}\subseteq K}\varepsilon_i\varepsilon_j.
\end{align}
Thus $\widehat{g}_K(\emptyset)=1$, the only nonzero non-constant Fourier coefficients are those indexed by two-element subsets of $K$, and \eqref{eq:cliquecoeff} follows. In particular, $\widehat{g}_K(S)=0$ for every $S\in\mathcal{F}_n$. It remains to verify nonnegativity. Since $k$ is odd, $S_K$ is an odd integer for every  $\varepsilon\in\{\pm1\}^n$, and hence $S_K^2\ge1$. Therefore $g_K(\varepsilon)=(S_K^2-1)/(k-1)\ge0$. Hence $g_K\in\mathcal{K}_n$.
\end{proof}

\noindent 
The parity of $k$ is essential: if $k$ is even, $S_K$ can vanish, in which case \eqref{eq:clique} gives $g_K=-1/(k-1)<0$. For odd $k$, the quantity $S_K$ measures the imbalance between the number of $+1$ and $-1$ values among the $\varepsilon_i$ with $i\in K$. More precisely, if $r$ of the $k$ signs are $+1$, then $S_K=2r-k$, and hence $g_K=\big((2r-k)^2-1\big)/(k-1)$. Thus $g_K$ is a normalized measure of the squared imbalance. Since $k$ is odd, the most balanced configurations have $|S_K|=1$, for which $g_K=0$, while $g_K$ increases as the imbalance grows. The case $k=3$ is particularly simple:
\begin{align*}
g_{K}=1+\varepsilon_1\varepsilon_2+\varepsilon_1\varepsilon_3+\varepsilon_2\varepsilon_3=4\cdot\mathbf{1}\bigl[x_1=x_2=x_3\bigr].
\end{align*}
Thus, for $k=3$, $g_K$ is four times the indicator of the event that $x$ is constant on $K$. In this case $|\widehat{g}_K(e)|=1$ for each of the three edges contained in $K$, so the corresponding two-bit parities are determined with certainty (see Fig.~\ref{fig3}). Since $2/(k-1)=1$ only for $k=3$, this is the unique odd-clique construction in which the edge bias attains its maximal possible value.

\begin{lemma}[Hitting probability]\label{lem:hit}
For $K\subseteq[n]$ with $|K|=k$, a uniformly random $\mathtt{M}\in\mathscr{M}_n$ contains no edge with both endpoints in $K$ with probability
\begin{align}\label{eq:miss}
q(n,k)=\prod_{j=1}^{k}\frac{n-k-j+1}{n-2j+1},
\end{align}
where $q(n,k)=0$ automatically when $2k>n$, since the factor with $j=n-k+1\le k$ then has vanishing numerator. Consequently, $\Phi(g_K)=\frac{2}{k-1}\bigl(1-q(n,k)\bigr)$. 
\end{lemma}

\begin{proof}
A matching avoids all edges inside $K$ exactly when each of the $k$ vertices of $K$ is matched to a distinct vertex outside $K$. This is impossible if $k>n-k$. Otherwise, the vertices of $K$ can be assigned distinct partners in $[n]\setminus K$ in $(n-k)(n-k-1)\cdots(n-2k+1)$ ways, while the remaining $n-2k$ vertices outside $K$ can be matched among themselves in $(n-2k-1)!!$ ways. Dividing by $|\mathscr{M}_n|=(n-1)!!$ and using $(n-1)!!/(n-2k-1)!!=\prod_{j=1}^{k}(n-2j+1)$ gives \eqref{eq:miss}. The formula for $\Phi(g_K)$ follows since by \eqref{eq:cliquecoeff} all nonzero coefficients are equal, so the maximum over $e\in \mathtt{M}$ is $2/(k-1)$ if $\mathtt{M}$ meets $K$ and $0$ otherwise.
\end{proof}

\begin{proposition}[Lower bound]\label{prop:lower}
For every even $n$,
\begin{align}\label{eq:lower}
R_n&\ge\max_{\substack{3\le k\le n\\ k\ \mathrm{odd}}}\tfrac{2}{k-1}\bigl(1-q(n,k)\bigr)\ge\tfrac{0.902}{\sqrt n}-O\left(\tfrac1n\right).
\end{align}
\end{proposition}
\begin{table}[t!]
\centering
\begin{tabular}{r|c|c|c|}
$n$ ~ &~ Best odd $k$ ~& $R_n\ge$ & $p^{\mathrm{NC}}_{\mathrm{opt}}\ge$\\\hline\hline
$6^\ast$ ~   & $3$  & $3/5=0.600000$    & ~~$0.800000$~~\\
$8^\ast$ ~   & $5$  & $1/2=0.500000$    & $0.750000$\\
$10$ ~ & $5$  & $0.436508$ & $0.718254$\\
$12$ ~  & $5$  & $0.378788$        & $0.689394$\\
$16$ ~ & $7$  & $0.303497$        & $0.651748$\\
$20$ ~ & $7$  &~~ $0.267286$~~& $0.633643$\\
$50$ ~ & $11$ & $0.151123$        & $0.575561$\\
$100$ ~ & $15$ & $0.101267$        & $0.550633$\\\hline
\end{tabular}
\caption{Best odd-clique values from Lemmas~\ref{lem:clique} and
\ref{lem:hit}. The corresponding quantum value is $1$ throughout. ($\ast$) Theorem~\ref{thm:exact} further shows for $n=6,8$ the corresponding values are in-fact optimal.}
\label{tab:values}
\end{table}

\begin{proof}
The first inequality follows immediately from Lemmas~\ref{lem:clique} and~\ref{lem:hit}. For the asymptotic lower bound, it suffices to consider odd $k$ with $2k\le n$. From \eqref{eq:miss}, 
\begin{align*}
\tfrac{n-k-j+1}{n-2j+1}=1-\tfrac{k-j}{n-2j+1}\le 1-\tfrac{k-j}{n},
\end{align*}
for $1\le j\le k$, and hence
\begin{align}\label{eq:missexp}
q(n,k)\le\prod_{i=0}^{k-1}\left(1-\tfrac{i}{n}\right)\le\exp\left(-\tfrac{k(k-1)}{2n}\right).
\end{align}
Writing $k-1=\alpha\sqrt n$ and using $k(k-1)\ge (k-1)^2$ gives
\begin{equation}
R_n \ge\tfrac{2}{\alpha\sqrt n}\left(1-\exp(-\alpha^2/2)\right).
\end{equation}
Maximizing $h(\alpha):=2/\alpha\big(1-\exp(-\alpha^2/2)\big)$ over $\alpha>0$ gives the stationary condition $\exp(\omega)=2\omega+1$ with $\omega=\alpha^2/2$, whose unique positive root is $\omega^\ast=1.256431\ldots$, thereby yielding 
\begin{align*}
\alpha^\ast=1.585201\ldots\text{ and } h(\alpha^\ast)=0.902512\ldots
\end{align*}
Rounding $\alpha^\ast\sqrt n+1$ to the nearest odd integer perturbs the value by
$O(1/n)$.
\end{proof}

Note that, for $k=3$, Eq.\eqref{eq:miss} gives $q(n,3)=(n-4)/(n-1)$ and hence $\Phi(g_{K_3})=3/(n-1)$. This is the performance obtained by certifying the parities within a single triple. Since a triangle is the largest class that can be certified with certainty by Lemma~\ref{lem:classsize}, it is natural to think that this strategy is optimal. It is not. At $n=8$, taking $k=5$ gives $\bigl|\widehat{g}_{K_5}(e)\bigr|=1/2$, and every perfect matching contains an edge entirely within $K_5$, because only $3$ vertices lie outside $K_5$. Thus, $\Phi(g_{K_5})=1/2>3/7=\Phi(g_{K_3})$ (see Table~\ref{tab:values}). Hence the optimal strategy need not certify any parity with certainty; instead, it can use weaker correlations spread over a larger set of vertices. The quintuple first outperforms the triangle at $n=8$ (an explicit protocol is detailed in Appendix \ref{app:n8}).

It remains to show that no strategy can outperform the odd clique strategies by more than a constant factor. The key constraint is a \emph{sum rule}: for every admissible ontic density \(g\), the squared pairwise correlations obey
\begin{equation}\label{eq:sumrulemotivation}
\sum_{i<j}\langle\varepsilon_i\varepsilon_j\rangle_g^2=O(1),
\end{equation}
with a constant independent of \(n\). Thus, although the pairwise correlations may be concentrated on a small number of pairs, their aggregate squared magnitude cannot grow with the number of pairs. The bound \eqref{eq:sumrulemotivation} follows from hypercontractivity of the noise operator on the Boolean hypercube, which converts the degree constraint on \(g\) into a norm bound~\cite{Bonami1970,Beckner1975,BenAroya2008,Montanaro2012}. A random perfect matching samples only \(n/2\) of the \(\binom n2\) possible pairs. Hence the sum rule severely limits the largest correlation that can occur on a matched edge. The following proposition makes this statement quantitative.

\begin{proposition}[Upper bound]\label{prop:upper}
For every even $n$,
\begin{align}\label{eq:upper}
R_n\le 2\sqrt{\tfrac{80}{n-1}}<\tfrac{18}{\sqrt{n-1}}.
\end{align}
\end{proposition}

\begin{proof}
(The idea) The proof involves two basic concepts:

\begin{enumerate}[itemsep=-.05cm, topsep=2pt, leftmargin=.5cm]
\item[$i.$] First, the total pairwise correlation carried by an admissible density is bounded by an absolute constant, independent of $n$:
\begin{align}\label{eq:appsumrule}
T:=\sum_{e\in\mathcal{E}_n}\widehat g(e)^2\;\le\;80 .
\end{align}
\item[$ii.$] Second, Bob's matching samples only $n/2$ of the $\binom n2$ pairs, so a fixed budget of correlation spread over $\Theta(n^2)$ edges leaves him, typically, an edge carrying correlation of order $n^{-1/2}$.
\end{enumerate}
The first statement is a \emph{sum rule}, proved by hypercontractivity; the second
is a counting argument. We provide the detailed proof in Appendix \ref{app:upper}.
\end{proof}

\subsection{The Inequality}

\noindent
Assembling the two bounds, obtained in Propositions~\ref{prop:lower} and~\ref{prop:upper}, gives the quantitative noncontextuality inequality promised at the beginning of the section, and also identifies the scale of the quantum violation that an experiment must resolve.

\begin{theorem}[Quantitative noncontextuality inequality]\label{thm:ncineq}
In any operational theory admitting a preparation-noncontextual hidden-variable model, the success probability in $\mathrm{PoHM}_n$ satisfies $p^{NC}_{\mathrm{succ}}(\mathrm{PoHM}_n)\le 1/2+9/\sqrt{n-1}$. Furthermore, this scaling is optimal up to a constant factor: $p_{\mathrm{succ}}^{\mathrm{NC}}(\mathrm{PoHM}_n)=
1/2+\Theta(n^{-1/2})$.
\end{theorem}

\noindent 
By contrast, the quantum protocol of Theorem~\ref{thm:qm} achieves $p_{\mathrm{succ}}(\mathrm{PoHM}_n)=1$ for all even $n$. The quantitative inequality of Theorem~\ref{thm:ncineq} may be viewed as the robust counterpart of the all-vs-nothing preparation-contextuality argument, in close analogy with the relation between the GHZ argument and Mermin's inequality~\cite{Mermin1990-1}. The GHZ argument establishes an ideal contradiction from perfect quantum correlations, while Mermin's inequality converts this contradiction into a quantitative condition that can be tested in the presence of experimental imperfections. Likewise, the possibilistic $\mathrm{PoHM}_n$ argument shows that perfect quantum success is incompatible with preparation noncontextuality, whereas Theorem~\ref{thm:ncineq} quantifies the maximal success probability of a preparation-noncontextual model. This places the two results in complementary roles: the former establishes the \emph{all-vs-nothing} contradiction, while the latter provides its experimentally testable, noise-tolerant counterpart.

It must be emphasized at once that the bound of Theorem~\ref{thm:ncineq} is nontrivial only for large $n$. Since $1/2+9/\sqrt{n-1}<1$ requires $\sqrt{n-1}>18$, the inequality carries content only for $n\ \ge\ 326$, and is vacuous for every $n$ appearing in Table~\ref{tab:values}. This is a limitation of the constant in Proposition~\ref{prop:upper}, not of the scaling, and it is confined to the asymptotic argument. As we now show, at the two smallest nontrivial values of $n$ the optimum can be determined exactly.

\subsection{Exact Values for $n=6$ and $n=8$}\label{ssec:exact}

\noindent 
The bound in Proposition~\ref{prop:upper} is asymptotic and its constant is not sharp. At the two smallest nontrivial values of $n$, however, the optimum can be determined exactly and coincides with the value achieved by the odd-clique strategies. The exact calculation reduces, through two elementary observations, to a finite optimization problem.

\begin{lemma}[Selection relaxation]\label{lem:topk}
For every $g\in\mathcal{K}_n$, let $|\,\widehat g(e_1)\,|\ge |\,\widehat g(e_2)\,|\ge\cdots$ be the pairwise Fourier coefficients arranged in non-increasing order. Then
\begin{align}\label{eq:topk}
\Phi(g)\le \frac{1}{n-1}\sum_{r=1}^{n-1}\bigl|\widehat g(e_r)\bigr|.
\end{align}
\end{lemma}

\begin{proof}
Write $m_e:=\bigl|\widehat g(e)\bigr|$, and let \(p_e\) be the probability that \(e\) is selected as a maximizing edge in~\eqref{eq:Phi} for a uniformly random \(\mathtt M\in\mathscr M_n\), with ties broken arbitrarily. Since exactly one edge is selected for each matching, $\sum_e p_e=1$. Moreover, an edge can be selected only when it belongs to the matching, and hence $0\le p_e\le \Pr[e\in\mathtt M]=1/(n-1)$. Therefore $\Phi(g)=\sum_e p_e m_e$. We may now relax the matching-induced constraints on the vector \(p\) and optimize only subject to $\sum_e p_e=1,~0\le p_e\le 1/(n-1)$. This linear program is maximized by assigning weight \(1/(n-1)\) to the \(n-1\) largest values of \(m_e\), and zero weight to all remaining edges. Hence, $\Phi(g)\le1/(n-1)\sum_{r=1}^{n-1}m_{e_r}$, where \(m_{e_1}\ge\cdots\ge m_{e_{|\mathcal E_n|}}\), proving~\eqref{eq:topk}.
\end{proof}

\begin{lemma}[Restriction]\label{lem:restrict}
Let $A\subseteq[n]$ and let $g\in\mathcal{K}_n$. Averaging $g$ over the coordinates outside $A$ yields $g_A(\varepsilon_A)=1+\sum_{e\subseteq A}\widehat g(e)\,\varepsilon_e$, and $g_A\in\mathcal{K}_{|A|}$. Conversely every $h\in\mathcal{K}_{|A|}$ extends to $\mathcal{K}_n$ by assigning zero to all edges meeting $[n]\setminus A$. Consequently, for any edge set $S$,
\begin{align}\label{eq:isoinv}
\max_{g\in\mathcal{K}_n}\sum_{e\in S}|\widehat g(e)|
\end{align}
depends only on the isomorphism class of $S$ viewed as an abstract graph.
\end{lemma}

\begin{proof}
Averaging over the coordinates in \([n]\setminus A\) preserves nonnegativity. In the Fourier expansion, this averaging projects onto the characters supported entirely on \(A\), so
\begin{align*}
\widehat{g_A}(\emptyset)=\widehat g(\emptyset)=1,\quad\widehat{g_A}(e)=\widehat g(e)\quad(e\subseteq A),
\end{align*}
with all other Fourier coefficients vanishing. Since \(g\) has degree at most two, it follows that \(g_A\in\mathcal K_{|A|}\). Conversely, if \(h\in\mathcal K_{|A|}\) is regarded as a function of \(\varepsilon_A\) alone, then its extension to \(\{\pm1\}^n\) is nonnegative because every assignment to \(\varepsilon_A\) extends to an assignment on \(\{\pm1\}^n\). Its Fourier coefficients outside \(A\) vanish, so the extension belongs to \(\mathcal K_n\). Taking \(A=V(S)\), we may therefore evaluate \eqref{eq:isoinv} entirely within \(\mathcal K_{|V(S)|}\). Finally, a permutation of the vertices induces a permutation of the coordinates and hence maps \(\mathcal K_{|V(S)|}\) onto itself while preserving the absolute values of the corresponding Fourier coefficients. Thus the quantity in~\eqref{eq:isoinv} depends only on the isomorphism class of the graph with edge set \(S\).
\end{proof}

\noindent 
Two further reductions substantially shrink the enumeration:
\begin{enumerate}[itemsep=-.02cm, topsep=2pt, leftmargin=.4cm]
\item[$\dagger$] First,
\begin{align*}
\sum_{e\in S}|\widehat g(e)|=\max_{s\in\{\pm1\}^{S}}\sum_{e\in S}s_e\,\widehat g(e).
\end{align*}
Thus, for each fixed pair \((S,s)\), the inner optimization in \eqref{eq:isoinv} is a linear program over the polytope \(\mathcal K_n\). The defining inequalities are \(g(\varepsilon)\ge0\) for \(\varepsilon\in\{\pm1\}^n\); since every \(g\in\mathcal K_n\) contains only even-degree Fourier terms, \(g(\varepsilon)=g(-\varepsilon)\), leaving \(2^{n-1}\) distinct facets.

\item[$\dagger\dagger$] Second, \(\mathcal K_n\) is invariant under the gauge transformation 
\begin{align*}
\widehat g(\{i,j\})\longmapsto\eta_i\eta_j\,\widehat g(\{i,j\}),\quad\eta\in\{\pm1\}^n,
\end{align*}
corresponding to the coordinate change \(\varepsilon_i\mapsto\eta_i\varepsilon_i\). On the sign patterns this acts by $s_{\{i,j\}}\longmapsto \eta_i\eta_j s_{\{i,j\}}$. Hence sign patterns on \(S\) are identified up to vertex switching. For a graph \(S\) with \(c(S)\) connected components, the number of switching-inequivalent sign patterns is $2^{\,|S|-|V(S)|+c(S)}$, namely \(2^{\beta_1(S)}\), where \(\beta_1(S)\) is the cyclomatic number of \(S\), rather than the naive \(2^{|S|}\).
\end{enumerate}

\begin{theorem}[Exact small-$n$ values]\label{thm:exact}
For the two smallest nontrivial even values of \(n\), $R_6=3/5$ and $R_8=1/2$. Consequently, $p^{\mathrm{NC}}_{\mathrm{opt}}(6)=4/5=0.80$ and $p^{\mathrm{NC}}_{\mathrm{opt}}(8)=3/4=0.75$. In both cases the optimum is attained by the odd-clique density: by a triangle for \(n=6\) and by a \(5\)-clique for \(n=8\).
\end{theorem}

\begin{proof}
The lower bounds are Lemmas~\ref{lem:clique} and~\ref{lem:hit}: $\Phi(g_{K_3})=3/5$ at $n=6$ and $\Phi(g_{K_5})=1/2$ at $n=8$. For the upper bounds we verify, by Lemma~\ref{lem:topk}, that
\begin{subequations}
\begin{align}
\max_{g\in\mathcal{K}_6}\ \Big(\text{sum of the $5$ largest }|\widehat g(e)|\Big)=3,\\
\max_{g\in\mathcal{K}_8}\ \Big(\text{sum of the $7$ largest }|\widehat g(e)|\Big)=\tfrac72 ,
\end{align}
\end{subequations}
whence $R_6\le3/5$ and $R_8\le1/2$. Each maximum is a finite maximum of linear programs, taken over pairs $(S,s)$ with $|S|=n-1$. For $n=6$ there are $\binom{15}{5}=3003$ edge sets; after gauge reduction this leaves $4890$ linear programs, each over the $32$ facets of $\mathcal{K}_6$, and the maximum of their values is exactly $3$, attained at a triangle carrying $|\widehat g(e)|=1$. For $n=8$ the $\binom{28}{7}=1\,184\,040$ edge sets fall, by Lemma~\ref{lem:restrict}, into $115$ isomorphism classes of $7$-edge graphs on at most $8$ vertices; after gauge reduction this leaves $295$ linear programs, and the maximum of their values is exactly $7/2$, attained on a graph spanning five vertices, with the optimal $g$ equal to $g_{K_5}$ up to gauge. The scripts performing these enumerations are provided as supplementary material.
\end{proof}

\noindent 
Theorem~\ref{thm:exact} is a computer-assisted but finite and exhaustive verification: no search heuristic is involved, and every pair $(S,s)$ is covered up to the two symmetries just described. The enumeration of isomorphism classes used in the $n=8$ case was validated independently by removing the bound on the number of vertices, whereupon it reproduces the known counts $2,5,11,26,68,177$ of graphs with $2,\dots,7$ edges and no isolated vertices. The linear programs were solved in floating-point arithmetic, so the reported optima are exact to within $10^{-9}$; a verification in exact rational arithmetic would remove this last caveat. 

The optimal values in Theorem~\ref{thm:exact} are not subject to the constant of Proposition~\ref{prop:upper}. They show that already at $n=6$, every preparation-noncontextual theory is bounded by $0.80$ while quantum theory attains $1$: an absolute gap of $0.200$, comparable in magnitude to the gap of $0.104$ between the noncontextual bound $3/4$ and the quantum value $\cos^2(\pi/8)$ in parity-oblivious multiplexing~\cite{Spekkens2009}. The quantitative separation is therefore experimentally accessible at the smallest nontrivial system size, independently of the asymptotic analysis.

Unlike Theorem~\ref{thm:pnc}, Theorem \ref{thm:ncineq} and more particularly Theorem \ref{thm:exact} admit, in principle, a statistical test: under the ideal parity-obliviousness condition, any observed success probability exceeding the noncontextual bound falsifies preparation noncontextuality, with the size of the violation quantified directly by the excess over the bound. Its experimental interpretation, however, requires separate control of the degree to which the implemented preparations satisfy the parity-obliviousness constraint. Importantly, the parity-obliviousness condition is an essential part of the test and must be experimentally certified rather than assumed. In contrast to the no-signaling condition in Bell experiments, it is not enforced by spacelike separation; it is a property of the implemented preparation procedures and therefore has to be established from observed statistics. In practice, the idealized requirement in Definition~\ref{def:po} presents two immediate difficulties: 
\begin{enumerate}[itemsep=-.02cm, topsep=-1pt, leftmargin=.4cm]
\item One cannot test equivalence against a literally unrestricted set of measurements. This issue can be addressed by tomography, or more generally by verifying the required operational equivalences on a tomographically complete set of measurements.

\item The exact equality required by Eq.~\eqref{eq:po} cannot be established experimentally from finite data. One can only bound the residual information about the forbidden parities that remains in the observed preparations. A fully robust test would therefore require a quantitative version of the noncontextuality inequality in which the upper bound on $p_{\mathrm{succ}}$ depends explicitly on such a leakage parameter. We leave this question as an interesting direction for future work.
\end{enumerate}

\medskip
\noindent 
Finally, we would like to point out that Propositions~\ref{prop:lower} and~\ref{prop:upper} leave a gap of roughly a factor of \(20\) in the constant pre-factor. The only loss in the present upper-bound argument appears in Step~1 of Proposition~\ref{prop:upper}. Indeed, for the odd-clique strategies, $T=\binom{k}{2}\times 4/(k-1)^2=2k/(k-1)\le 3$, whereas the general hypercontractive argument gives the much weaker bound \(T\le80\). This motivates the conjecture $T\le3~\text{for all }g\in\mathcal K_n$. If true, it will improve the constant in Theorem~\ref{thm:ncineq} from \(9\) to \(\sqrt3\). It remains open whether the odd-clique strategies are globally optimal for \(\Phi\). They are optimal within the rank-one family $g\propto (w\cdot\varepsilon)^2-\delta^2$, but a characterization of the extreme rays of the cone \(\mathcal K_n\) is not currently known. Establishing such a characterization, or otherwise proving the conjectured bound \(T\le3\), could therefore substantially sharpen the quantitative noncontextuality inequality.

\section{CONCLUSION AND OUTLOOK}\label{sec:discussion}

\noindent
We have introduced a communication task, parity-oblivious hidden matching, in which quantum theory succeeds with certainty while no preparation-noncontextual theory does, for any even $n\ge6$. The contradiction is possibilistic rather than statistical, and in this respect stands to the parity-oblivious multiplexing of Ref.~\cite{Spekkens2009} as the Greenberger--Horne--Zeilinger argument stands to the Clauser--Horne--Shimony--Holt inequality. We have also determined how well a noncontextual theory can do, exactly at $n=6$ and $n=8$ and to within a constant factor in general.

\medskip
\noindent{\bf Relation to prior works}

\smallskip
\noindent
Possibilistic reasoning is not new in this setting, but it is important to distinguish the present use from earlier ones. Simmons \emph{et al.}~\cite{Simmons2017} weaken the noncontextuality \emph{assumption}, requiring only that preparations with the same set of operationally possible events have ontic representations with the same support. They then use a Hardy-type argument to show that quantum theory is preparation contextual even under this weaker hypothesis. Thus, their result strengthens the Spekkens no-go theorem by weakening its premise, and is a structural statement about ontological models, with no information-processing task attached.

Here we retain the standard preparation-noncontextuality assumption (Definition~\ref{def:pnc}) and instead attach contextuality to a concrete task: quantum theory achieves unit success probability in $\mathrm{PoHM}_n$, whereas no preparation-noncontextual theory can, with the bound already reduced to $0.750$ at $n=8$. To our knowledge, no earlier task-based demonstration of preparation contextuality achieves $p_{\mathrm{succ}}=1$ on the quantum side.

A sheaf-theoretic notion of preparation contextuality, with a possibilistic formulation as a strengthening, was recently introduced by Williams \emph{et al.}~\cite{Williams2026}. Their framework is explicitly distinct from the operational approach of Spekkens and formulates preparation contextuality as an obstruction to a global stochastic extension; in their PBR-type example, they prove the stronger nonexistence of a possibilistic global response matrix. At a more structural level, Shahandeh~\cite{Shahandeh2021} characterizes the general probabilistic theories that admit noncontextual ontological models, showing that, under the no-restriction hypothesis, these are precisely the simplicial theories. This addresses a complementary question: when is a theory noncontextual at all? Our concern instead is operational and task-oriented—what advantage can preparation contextuality provide in a concrete information-processing task?

\medskip
\noindent{\bf Relation to communication complexity and to Bell nonlocality} 

\smallskip
\noindent The hidden matching problem was introduced by Bar-Yossef, Jayram and Kerenidis to separate quantum from classical one-way communication complexity~\cite{BarYossef2008}, and its nonlocal version was analysed by Buhrman, Regev, Scarpa and de Wolf~\cite{Buhrman2012}. The parallel with the present work is close enough to be worth spelling out, since the three settings constrain three different resources while producing
the same exponent.
 
In the communication setting, Ref.~\cite{Buhrman2012} shows that a classical one-way protocol in which Alice sends $c$ bits wins hidden matching with probability at most $1/2+(c+1)/\sqrt{n-1}$ by an argument that bounds, via the Kahn--Kalai--Linial inequality~\cite{Kahn1988}, the number of pairs $\{i,j\}$ whose parity a short message can predict well. In the nonlocal setting they obtain a two-player game whose entangled value is $1$, using $\log n$ EPR pairs, while every classical strategy wins with probability differing from $1/2$ by at most $O\bigl(\log n/\sqrt n\bigr)$; the logarithmic factor arises because their game is analyzed by reduction to a
communication protocol of $\log n$ bits.
 
Our Theorem~\ref{thm:ncineq} has exactly the form of the communication bound of Buhrman {\it et al} with $c+1=9$, and this is the more illuminating way to read it. Parity obliviousness places no bound on the size of the ontic state space, nor on how much information $\lambda$ carries about $x$ as a whole; yet for the purposes of this task, a preparation-noncontextual ontic state is worth no more than a constant number of bits of communication. If the conjecture $T\le3$ holds, the constant improves to $\sqrt3$, that is, to less than one bit. The mechanism behind the bound is the same in both settings---a hypercontractive inequality on the Boolean hypercube converts a restriction on Alice's message into a bound on the pairwise correlations it can carry---but the restriction itself is of a different kind: a bound on message \emph{length} in Ref.~\cite{Buhrman2012}, a bound on the \emph{content} of a message of unbounded length here. Note also that our bound carries no logarithmic factor, since the ontic state is used directly as the message and no reduction through a $\log n$-bit protocol is needed.
 
The $n^{-1/2}$ exponent common to all three settings reflects a single combinatorial feature of hidden matching: correlations distributed over $\binom n2$ pairs must be accessed through a matching containing only $n/2$ of them. What differs is the resource being limited. The communication bound of Ref.~\cite{BarYossef2008} restricts the number of classical bits Alice may transmit while imposing no condition on what those bits contain; our result restricts what the message may contain while imposing no condition on its length. Accordingly, Theorem~\ref{thm:ncineq} does not lower-bound the communication required by a noncontextual protocol, but upper-bounds the success bias achievable by any preparation-noncontextual model meeting the obliviousness constraint.
 
\medskip
\noindent{\bf Experimental prospects} 

\smallskip
\noindent Two features make $\mathrm{PoHM}_n$ comparatively economical to implement. First, as observed in Ref.~\cite{Buhrman2012} in the nonlocal context, a task that quantum theory wins \emph{with certainty} is experimentally attractive: the quantum value does not shrink with the system size, so the signal to be resolved is the noncontextual bound falling away from $1$, rather than a small success probability that must be accumulated over many runs. Second, by Theorem~\ref{thm:exact} the separation is already quantitative and nonvacuous at $n=6$: every preparation-noncontextual theory is bounded by $0.80$, against a quantum value of $1$. The required system has dimension six, the $64$ preparations are the real equal-amplitude states $\ket{\psi_x}$, which form a single orbit under a group of diagonal sign flips, and each measurement is a product of $50{:}50$ operations on disjoint pairs of modes; neither entanglement nor joint measurements are needed.
 
The number of measurement settings can be reduced substantially. Bob's input ranges a priori over all $(n-1)!!$ perfect matchings, which is $15$ already at $n=6$. Following the restriction used in Ref.~\cite{Buhrman2012}, consider instead the $n/2$ matchings
\begin{align}\label{eq:cyclicfam}
\scriptstyle\mathtt{M}_k:=\Bigl\{\bigl\{i,\tfrac n2+1+\bigl((i+k-1)\bmod \tfrac n2\bigr)\bigr\}\, :\, 1\le i\le \tfrac n2\Bigr\},
\end{align}
for $k\in\{0,\dots,n/2-1\}$. We have verified for $n=6,8,10,12$ that Theorem~\ref{thm:cl} continues to hold when Bob's input is drawn from this family alone: no single clique of size at most three meets every $\mathtt{M}_k$, which by Lemma~\ref{lem:classsize}(ii) is what zero error would require. The quantum protocol of Theorem~\ref{thm:qm} is unaffected, since it succeeds for every matching. Bob's settings at $n=6$ therefore drop from $15$ to $3$. The bounded-error analysis also survives, since Lemma~\ref{lem:topk} uses only the value of $\Pr[e\in\mathtt{M}]$, which becomes $2/n$ for the pairs occurring in~\eqref{eq:cyclicfam} and zero otherwise; the exponent $n^{-1/2}$ is unchanged and only the constant differs.
 
The exact values behave differently in the two cases, and the computation is short enough to carry out in full: with $n/2$ matchings, $\Phi$ is a maximum over selections of one $(\text{edge},\text{sign})$ per matching, so $R_n$ for the restricted family is the maximum of $n^{n/2}$ linear programs over $\mathcal{K}_n$---$216$ at $n=6$ and $4096$ at $n=8$---with no relaxation needed. Writing $R_n^{\mathrm{cyc}}$ for the resulting optimum, we find
\begin{align}\label{eq:cycvalues}
R_6^{\mathrm{cyc}}=\tfrac23>\tfrac35=R_6,
\quad
R_8^{\mathrm{cyc}}=\tfrac12=R_8,
\end{align}
so the noncontextual bound rises from $0.800$ to $0.833$ at $n=6$, but is unchanged at $0.750$ for $n=8$. The reason for the loss at $n=6$ is visible in the optimizer: every matching in~\eqref{eq:cyclicfam} is bipartite between $\{1,\dots,\tfrac n2\}$ and $\{\tfrac n2+1,\dots,n\}$, so the triangle $\{1,2,3\}$, which is optimal for the full family, meets none of them; the best restricted strategy is instead a triangle straddling the bipartition, which meets two of the three matchings and scores $2/3$. Shrinking the family can only help a noncontextual strategy, since it has fewer matchings to cover.
 
For a quantitative experiment the case $n=8$ is therefore the attractive one: the noncontextual bound $0.750$ against a quantum value of $1$ is retained while Bob's settings fall from $105$ to $4$. At $n=6$ one must choose between the sharper bound $0.800$ with $15$ settings and the weaker bound $0.833$ with $3$. The all-vs-nothing statement, Theorem~\ref{thm:cl}, holds in every case.
 
The principal experimental difficulty is not the success probability but the certification of parity obliviousness. Unlike no-signaling in a Bell test, it is not enforced by spacelike separation: it is a property of the implemented preparations and must be established from observed statistics, for the $2^n-1-\binom n2$ subsets in $\mathcal{F}_n$. The two idealizations already noted---that one cannot test equivalence against all measurements, and that exact equalities cannot be certified from finite data---are the same obstacles confronted in the parity-oblivious multiplexing experiment of Ref.~\cite{Spekkens2009}, and were addressed in a more systematic way in~\cite{Mazurek2016}, which constructs secondary preparation procedures satisfying the required operational equivalences exactly by construction. Adapting that method here would remove the need for any continuity assumption, and is the natural next step towards a decisive experiment.

\medskip
\noindent{\bf Open questions}

\smallskip
\noindent Several questions are left open. The odd-clique strategies satisfy $T=2k/(k-1)\le3$, whereas hypercontractivity yields only $T\le80$. Proving $T\le3$ for every $g\in\mathcal{K}_n$ would sharpen Theorem~\ref{thm:ncineq} from $9$ to $\sqrt3$ and, more importantly, would make the asymptotic inequality non-vacuous from $n\ge14$ rather than $n\ge326$. By Theorem~\ref{thm:exact} the odd cliques are optimal at $n=6$ and $n=8$, and they are optimal within the rank-one family $g\propto(w\cdot\varepsilon)^2-\delta^2$ for every $n$. Whether they are globally optimal is open; a characterization of the extreme rays of $\mathcal{K}_n$, the polar dual of the cut polytope, would settle it. The enumeration used in Theorem~\ref{thm:exact} does not scale, since the number of isomorphism classes of $(n-1)$-edge graphs grows rapidly. A bound on $p_{\rm succ}$ depending explicitly on a measure of residual information about the forbidden parities would convert Theorem~\ref{thm:ncineq} into a statement directly applicable to imperfect data. The family $\mathcal{F}_n$ was chosen so that the permitted information is exactly what the task requires. Remark~\ref{rem:even} shows that weakening it to odd-weight subsets destroys the impossibility entirely; imposing obliviousness only for $|S|=2k$ with $k\ge3$ raises the threshold from $n\ge6$ to $n\ge4k-2$. Mapping out which obliviousness constraints yield all-vs-nothing contradictions, and which yield only quantitative ones, would clarify what makes a task contextuality-powered.

\appendix

\section{An explicit protocol for $n=8$}\label{app:n8}

\noindent
Theorem~\ref{thm:exact} gives \(R_8=1/2\), hence \(p_{\mathrm{opt}}^{\mathrm{NC}}(8)=3/4\), attained by the odd-clique density \(g_K\) with \(|K|=5\). Although the reduction of Proposition~\ref{prop:reduction} yields a classical encoding with \(2^8\) messages, the same optimum admits a much simpler realization using only four bits, with decoding given by the identity. The aim of this appendix is to provide an explicit protocol, which is closely analogous to the majority encoding for the classical \(3\mapsto1\) random-access code.

\begin{figure*}[t!]
\centering
\includegraphics[width=1\linewidth]{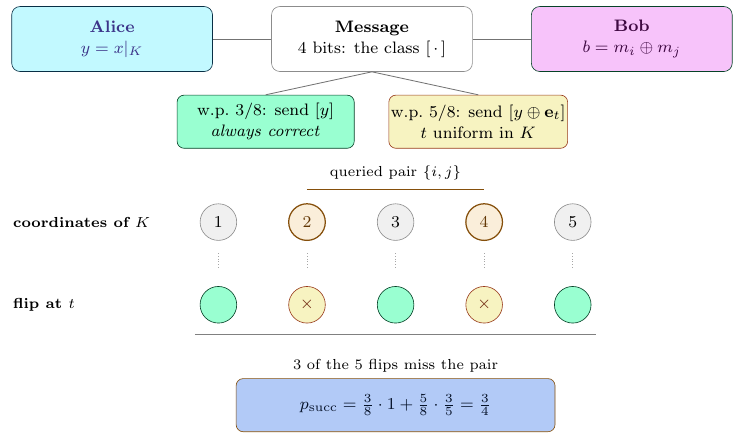}
\caption{(Color online) The optimal preparation-noncontextual protocol for $\mathrm{PoHM}_8$, drawn for the queried pair $\{i,j\}=\{2,4\}$. Alice transmits the five bits $y=x|_K$ of her input on the clique $K=\{1,\dots,5\}$, modulo complementation, so that the message is the four bits $d_t=y_1\oplus y_t$; with probability $5/8$ a single uniformly chosen coordinate $t$ is flipped first. Bob selects an edge $\{i,j\}\in\mathtt{M}$ with $i,j\in K$---one always exists, since only three vertices lie outside $K$---and outputs $b=m_i\oplus m_j$. A flip at $t$ leaves this parity intact precisely when $t\notin\{i,j\}$ (\checkmark), which happens for three of the five coordinates, and spoils it otherwise ($\times$). Note that the coordinates $x_6,x_7,x_8$ play no role: the message depends on $x|_K$ alone.}\vspace{0cm}
\label{fig4}
\end{figure*}

Throughout, this appendix \(K=\{1,2,3,4,5\}\), and we write \(y=(x_1,\ldots,x_5)\) for Alice's input restricted to \(K\). For \(z\in\{0,1\}^5\), let \(\bar z:=z\oplus1^5\) and define the complement class
\begin{align}
[z]:=\{z,\bar z\}.
\end{align}
This class is uniquely specified by the four bits
\begin{align}\label{eq:appdbits}
d_t:=z_1\oplus z_t,\qquad t=2,\ldots,5,
\end{align}
which are invariant under \(z\mapsto\bar z\). Hence there are exactly \(2^4=16\) distinct messages. We identify the coordinates of a five-bit string with the vertices of \(K\), so that \(z_i\) denotes the bit associated with vertex \(i\in K\).

\medskip
\noindent{\bf The protocol:} 

\medskip
\noindent
{\bf Alice.} With probability $3/8$ sends $[y]$; with probability $5/8$, chooses $t\in\{1,\dots,5\}$ uniformly at random and sends $[\,y\oplus \mathbf{e}_t\,]$, where $\mathbf{e}_t$ is the $t$-th standard basis vector.

\medskip
\noindent
\textbf{Bob.} Given the matching $\mathtt{M}$, chooses an edge $\{i,j\}\in\mathtt{M}$ with $i,j\in K$ and output that edge together with
\begin{align}\label{eq:appdecode}
b=m_i\oplus m_j ,
\end{align}
where $m=(m_1,\dots,m_5)\in\{0,1\}^5$ is either representative of the received class, its coordinates indexed by the vertices of $K$. Equivalently, in terms of the four bits \eqref{eq:appdbits} actually transmitted,
\begin{align}\label{eq:appdecode2}
b=\begin{cases}
d_j, & i=1,\\[2pt]
d_i\oplus d_j, & i,j\ge2 .
\end{cases}
\end{align}

\noindent
Note that, Bob can always find such an edge: only $3$ vertices lie outside $K$, so at most three of the five vertices of $K$ are matched outward and at least two are matched to each other (see Fig.~\ref{fig4}). The decoding \eqref{eq:appdecode} does not depend on the choice of representative, since $\bar m_i\oplus\bar m_j=(1\oplus m_i)\oplus(1\oplus m_j)=m_i\oplus m_j$. The two forms \eqref{eq:appdecode} and \eqref{eq:appdecode2} agree, because $d_i\oplus d_j=(z_1\oplus z_i)\oplus(z_1\oplus z_j)=z_i\oplus z_j$. Note that Bob never needs, and cannot obtain, any individual bit of $y$: the transmitted message determines pairwise parities only, which is what makes the encoding oblivious to the parities of odd order. And Alice's distribution is normalized, the two cases contributing $3/8$ and $5\times 1/8$.

\begin{table}[t!]
\centering
\begin{tabular}{c|c|c|c|c|}
\hline
~~~$d$~~~ & ~~~$S_K$~~~ & ~~~$g_{K}$~~~ & ~~~\#$\{\eta|_K\}$~~~ & ~~~weight~~~\\
\hline\hline
$0$ & $+5$ & $6$ & $1$  & $6$\\
$1$ & $+3$ & $2$ & $5$  & $10$\\
$2$ & $+1$ & $0$ & $10$ & $0$\\
$3$ & $-1$ & $0$ & $10$ & $0$\\
$4$ & $-3$ & $2$ & $5$  & $10$\\
$5$ & $-5$ & $6$ & $1$  & $6$\\
\hline
    &      &     & $32$ & $32$\\
\hline
\end{tabular}
\caption{The weight $g_{K}=(S_K^2-1)/4$ assigned to a message $\eta$ at Hamming distance $d$ from $y=x|_K$, with $S_K=5-2d$. Only $d\in\{0,1,4,5\}$ carries weight. The two columns on the right give the number of such $\eta|_K$ and their total weight; the weights sum to $2^5$, as they must.}
\label{tab:app1}
\end{table}

\begin{table}[b]
\centering
\begin{tabular}{c|c|c}
\hline
 & $3\mapsto1$ RAC & $\mathrm{PoHM}_8$ with $K_5$\\
\hline
Alice encodes & $\mathrm{maj}(x_1,x_2,x_3)$& $[y]$ w.p. $3/8$, $[y\oplus\mathbf{e}_t]$ w.p. $5/8$\\
message size  & $1$ bit & $4$ bits\\
Bob decodes   & identity & identity: $b=m_i\oplus m_j$\\
value         & $3/4$ & $3/4$\\
\hline
\end{tabular}
 \caption{The optimal noncontextual strategy for $\mathrm{PoHM}_8$ beside the majority encoding for the classical $3\mapsto1$ random access code~\cite{Ambainis2002}. Both are majority-type encodings with identity decoding and both attain $3/4$; what differs is that Bob reads off a pairwise parity rather than a single bit, which is why the noise takes the form of a bit flip in a five-bit block rather than a majority vote over three.}
\label{tab:app2}
\end{table}

\begin{fact}\label{fact:appvalue}
The protocol is parity-oblivious and succeeds with probability exactly \(3/4\). By Theorem~\ref{thm:exact}, this is the optimal success probability among preparation-noncontextual strategies.
\end{fact}

\begin{proof}
Fix \(\mathtt M\), and let \(\{i,j\}\subseteq K\) be the edge selected by Bob. If Alice sends \([y]\), then $m_i\oplus m_j=y_i\oplus y_j=x_i\oplus x_j$, and hence Bob is correct. If she sends \([y\oplus \mathbf{e}_t]\), then, for either choice of representative \(m\), the parity \(m_i\oplus m_j\) differs from \(y_i\oplus y_j\) exactly when \(t\in\{i,j\}\). Thus Bob is correct for the \(3\) values of \(t\notin\{i,j\}\), out of the \(5\) equally likely choices. Therefore
\begin{align}\label{eq:appsuccess}
p_{\mathrm{succ}}=\frac38\cdot 1+\frac58\cdot\frac35=\frac34.
\end{align}

\medskip
\noindent 
Obliviousness is inherited from the encoding of
Proposition~\ref{prop:reduction}, since the present protocol is obtained from it by successive coarse-grainings. The first coarse-graining is the restriction map $q_1:\eta\longmapsto z:=\eta|_K\in\{0,1\}^5$. The second is the quotient map $q_2: z\longmapsto [z]:=\{z,\bar z\}$.

The clique density $g_K=(S_K^2-1)/4$ depends on the argument only through the Hamming distance $d$ between $\eta|_K$ and $y$, via $S_K=5-2d$. The three coordinates outside \(K\) therefore play no role. As shown in Table~\ref{tab:app1}, the total weight is \(32=2^5\), while \(g_K\) vanishes for the \(20\) strings at Hamming distances \(2\) and \(3\). Thus only the cases \(d=0,1,4,5\) occur. The first coarse-graining does not affect decoding because Bob chooses an edge \(\{i,j\}\subseteq K\) and uses only the parity \(\eta_i\oplus\eta_j=z_i\oplus z_j\), which is unchanged by discarding the coordinates outside \(K\). The second coarse-graining is likewise lossless for the task, since $\bar z_i\oplus\bar z_j=(1\oplus z_i)\oplus(1\oplus z_j)=z_i\oplus z_j$. Thus \(z\) and \(\bar z\) are indistinguishable to Bob and may be merged into a single message \([z]\). This identifies the rows \(d\) and \(5-d\) of Table~\ref{tab:app1}. Consequently, the induced distribution on the \(16\) complement classes is
\begin{align}
\Pr([y])=\frac{6+6}{32}=\frac38,\quad
\Pr([y\oplus \mathbf{e}_t])=\frac{2+2}{32}=\frac18,
\end{align}
for $t\in\{1,\ldots,5\}$, which is exactly the explicit four-bit protocol given above. Finally, parity-obliviousness is preserved under both coarse-grainings: if the original sampled string \(\eta\) is independent of every forbidden parity, then any deterministic function of \(\eta\), including
\(q_2\circ q_1(\eta)\), is independent of the same parity.
\end{proof}

\smallskip
\noindent {\bf Comparison with the random access code:} 
The protocol is a majority encoding, in the sense of the footnote in Section~\ref{ssec:beststrategy}: by \eqref{eq:clique} the weight $g_{K}$ is an increasing function of the imbalance $|S_K|$ of the five bits of $K$, so Alice's message is biased towards the more imbalanced configurations, and Bob decodes by the identity. Table~\ref{tab:app2} sets the two side by side.

\medskip
\noindent It is worth noting that the message here consists of $4=n/2$ bits, which by Corollary~\ref{coro:n/2optimal} is exactly the zero-error classical one-way communication complexity of the \emph{unconstrained} task $\mathrm{HM}_8$. The same communication budget therefore buys certainty without the obliviousness constraint and only $3/4$ with it; the
difference is precisely what parity obliviousness costs.

\section{Proof of the upper bound}\label{app:upper}

\noindent
This appendix gives the proof of
Proposition~\ref{prop:upper}: for every even $n$,
\begin{align}\label{eq:appupper}
R_n\;\le\;2\sqrt{\frac{80}{n-1}}\;<\;\frac{18}{\sqrt{n-1}} .
\end{align}
The argument is simple. But, here we detail the tools for the sake of completeness. 

\medskip
\noindent{\bf Preliminaries: norms and Parseval}\\

\noindent
For $g:\{\pm1\}^n\to\mathbb{R}$ and $p\ge1$ we write
\begin{align}\label{eq:appnorm}
\|g\|_p:=\Bigl(\mathbb{E}_\varepsilon\bigl|g(\varepsilon)\bigr|^p\Bigr)^{1/p},
\end{align}
the expectation being over the uniform distribution on $\{\pm1\}^n$. The normalized constant function $1$ has $\|1\|_p=1$ for every $p$, and $\|g\|_p$ is non-decreasing in $p$. Recall the \emph{Parseval's identity}: since the characters $\{\varepsilon_S\}_{S\subseteq[n]}$ form an orthonormal basis of the space of real functions on $\{\pm1\}^n$ with respect to the inner product $\langle f,h\rangle=\mathbb{E}[fh]$,
\begin{align}\label{eq:appparseval}
\|g\|_2^2=\mathbb{E}\bigl[g^2\bigr]=\sum_{S\subseteq[n]}\widehat g(S)^2 .
\end{align}
Furthermore, for $g\ge0$ we have $\|g\|_1=\mathbb{E}[|g|]=\mathbb{E}[g]=\widehat g(\emptyset)$, which for $g\in\mathcal{K}_n$ equals $1$.

\medskip
\noindent{\bf The sum rule \texorpdfstring{$T\le80$}{T <= 80}}

\noindent The content here is that an admissible density cannot carry an unbounded amount of pairwise correlation, however it is distributed. The tool is hypercontractivity inequality of Bonami~\cite{Bonami1970} and Beckner~\cite{Beckner1975} (see also~\cite{deWolf2008}). For a function $g$ of \emph{degree at most $d$}, meaning that $\widehat g(S)=0$ whenever $|S|>d$,
\begin{align}\label{eq:appbonami}
\|g\|_4\;\le\;3^{d/2}\,\|g\|_2 .
\end{align}
Every $g\in\mathcal{K}_n$ has degree at most $2$ by Definition~\ref{def:admissible}, so \eqref{eq:appbonami} gives
\begin{align}\label{eq:appbonami2}
\|g\|_4\;\le\;3\,\|g\|_2 .
\end{align}
Physically, \eqref{eq:appbonami2} says that a low-degree function cannot be too spiky: its fourth moment is controlled by its second. A function concentrated on a very small set would have $\|g\|_4$ far larger than $\|g\|_2$, and the inequality forbids this once the degree is bounded.

We now apply \eqref{eq:appbonami2} against the normalization $\|g\|_1=1$. The two are connected by an interpolation inequality: for $1\le p\le r\le q$ and $\theta\in[0,1]$ with $1/r=\theta/p+(1-\theta)/q$,
\begin{align}\label{eq:appinterp}
\|g\|_r\;\le\;\|g\|_p^{\theta}\,\|g\|_q^{1-\theta} .
\end{align}
This is the statement that $p\mapsto\log\|g\|_{1/p}$ is convex, and it follows from H\"older's inequality: writing $|g|^{r}=|g|^{\theta r}\,|g|^{(1-\theta)r}$ and applying H\"older identity with the exponents $u=p/(\theta r)$ and $v=q/((1-\theta)r)$, which are conjugate because
\begin{align}
\frac1u+\frac1v=\frac{\theta r}{p}+\frac{(1-\theta)r}{q}=1 ,
\end{align}
one obtains $\|g\|_r^{r}\le\|g\|_p^{\theta r}\|g\|_q^{(1-\theta)r}$, and taking
$r$-th roots gives \eqref{eq:appinterp}.
Taking $r=2$, $p=1$ and $q=4$ implies $\theta=1/3$; so that
\begin{align}\label{eq:appinterp2}
\|g\|_2\;\le\;\|g\|_1^{1/3}\,\|g\|_4^{2/3} .
\end{align}
Since $g\ge0$ and $\widehat g(\emptyset)=1$ we have $\|g\|_1=1$, and the
first factor drops out. Substituting \eqref{eq:appbonami2} into
\eqref{eq:appinterp2},
\begin{align}\label{eq:appchain}
\|g\|_2\;\le\;\|g\|_4^{2/3}\;\le\;\bigl(3\|g\|_2\bigr)^{2/3}
=3^{2/3}\,\|g\|_2^{2/3} .
\end{align}
If $\|g\|_2=0$ then $g\equiv0$, which is excluded by
$\widehat g(\emptyset)=1$; so we may divide \eqref{eq:appchain} by
$\|g\|_2^{2/3}>0$ to obtain $\|g\|_2^{1/3}\le3^{2/3}$, that is,
\begin{align}\label{eq:appl2}
\|g\|_2\;\le\;3^{2}=9 .
\end{align}

Finally we convert \eqref{eq:appl2} into a statement about the pairwise coefficients. By Parseval \eqref{eq:appparseval},
\begin{align}\label{eq:appT}
T&=\sum_{|S|=2}\widehat g(S)^2\nonumber\\
&=\|g\|_2^2-\widehat g(\emptyset)^2-\sum_{i}\widehat g(\{i\})^2
-\sum_{|S|\ge3}\widehat g(S)^2 .
\end{align}
Two summation terms vanishes as $g$ has degree at most two and $\{i\}\in\mathcal{F}_n$; and $\widehat g(\emptyset)^2=1$. Hence
\begin{align}\label{eq:appTfinal}
T\;\le\;81-1\;=\;80 ,
\end{align}
Note that the bound is independent of $n$. This will be exploited in the rest of the argument.

\medskip
\noindent{\bf From the sum rule to a bound on \texorpdfstring{$\Phi$}{Phi}}

\noindent Write $b_e:=\widehat g(e)$, so that $\sum_e b_e^2=T\le80$. We must bound
\begin{align}
\Phi(g)=\mathbb{E}_{\mathtt{M}}\Bigl[\max_{e\in\mathtt{M}}|b_e|\Bigr].
\end{align}

\noindent (a) \emph{Few edges can be strongly correlated.} For a threshold $\theta\in(0,1)$ let
\begin{align}\label{eq:appGtheta}
G_\theta:=\bigl\{e\in\mathcal{E}_n:|b_e|>\theta\bigr\}
\end{align}
be the set of edges carrying correlation above $\theta$. Every such edge contributes more than $\theta^2$ to $T$, so
\begin{align}\label{eq:appmarkov}
\theta^2\,|G_\theta|\;<\;\sum_{e\in G_\theta}b_e^2\;\le\;T
\quad\Longrightarrow\quad
|G_\theta|\;\le\;T/\theta^{2} .
\end{align}
This is Markov's inequality applied to the numbers $b_e^2$; the sum rule
is what makes the right-hand side independent of $n$.

\smallskip
\noindent (b)\emph{A random matching rarely meets a small edge set.} A fixed edge $e=\{i,j\}$ belongs to a uniformly random perfect matching $\mathtt{M}\in\mathscr{M}_n$ with probability
\begin{align}\label{eq:appedgeprob}
\Pr[e\in\mathtt{M}]=\frac{(n-3)!!}{(n-1)!!}=\frac{1}{n-1}.
\end{align}
Therefore, by linearity of expectation followed by Markov's inequality in the form $\Pr[X\ge1]\le\mathbb{E}[X]$ for the nonnegative integer-valued $X=|\mathtt{M}\cap G_\theta|$,
\begin{align}\label{eq:appunion}
\Pr_{\mathtt{M}}\bigl[\mathtt{M}\cap G_\theta\neq\emptyset\bigr]\;&\le\;\mathbb{E}_{\mathtt{M}}\bigl|\mathtt{M}\cap G_\theta\bigr|=\sum_{e\in G_\theta}\Pr[e\in\mathtt{M}]\nonumber\\
&=\frac{|G_\theta|}{n-1}
\;\le\;\frac{T}{(n-1)\,\theta^{2}} .
\end{align}
The factor $1/(n-1)$ is where the sparsity of a matching enters: it
samples only $n/2$ of the $\binom n2$ edges.

\smallskip
\noindent (c)\emph{Layer-cake integration.} For a nonnegative random variable $X$ bounded by $1$ one has the elementary identity $\mathbb{E}[X]=\int_0^1\Pr[X>\theta]\,d\theta$, obtained by writing $X=\int_0^1\mathbf 1[X>\theta]\,d\theta$ and exchanging expectation with the integral. Applying this to $X=\max_{e\in\mathtt{M}}|b_e|$, which lies in $[0,1]$ by \eqref{eq:coeffbound}, and noting that $\{\max_{e\in\mathtt{M}}|b_e|>\theta\}$ is exactly the event $\{\mathtt{M}\cap G_\theta\neq\emptyset\}$,
\begin{align}\label{eq:applayer}
\Phi(g)&=\int_0^1
\Pr_{\mathtt{M}}\bigl[\mathtt{M}\cap G_\theta\neq\emptyset\bigr]\,d\theta\nonumber\\
&\le\;\int_0^1\min\Bigl(1,\ \frac{T}{(n-1)\theta^{2}}\Bigr)\,d\theta ,
\end{align}
where the trivial bound $\Pr\le1$ has been used alongside
\eqref{eq:appunion}. Now, putting $\theta_0:=\sqrt{T/(n-1)}$, the threshold at which the two competing bounds in \eqref{eq:applayer}
cross. If $\theta_0\ge1$ then $2\theta_0\ge2>1\ge\Phi(g)$ and
\eqref{eq:appupper} holds trivially, so assume $\theta_0<1$. Splitting
the integral at $\theta_0$ and using the bound $1$ below and
$T/((n-1)\theta^2)$ above,
\begin{align}\label{eq:appsplit}
\Phi(g)
&\le\int_0^{\theta_0}1\,d\theta
+\int_{\theta_0}^{1}\frac{T}{(n-1)\theta^{2}}\,d\theta\nonumber\\
&\le\;\theta_0+\int_{\theta_0}^{\infty}\frac{T}{(n-1)\theta^{2}}\,d\theta
\nonumber\\[2pt]
&=\theta_0+\frac{T}{n-1}\cdot\frac{1}{\theta_0}
=\theta_0+\frac{\theta_0^{2}}{\theta_0}
=2\theta_0 ,
\end{align}
where the second line used $\int_{\theta_0}^\infty\theta^{-2}d\theta
=1/\theta_0$. Combining \eqref{eq:appsplit} with the sum rule
\eqref{eq:appTfinal} we obtain,
\begin{align}
\Phi(g)\;\le\;2\sqrt{\frac{T}{n-1}}\;\le\;2\sqrt{\frac{80}{n-1}} .
\end{align}
As $g\in\mathcal{K}_n$ was arbitrary, the same bound holds for
$R_n=\max_{g\in\mathcal{K}_n}\Phi(g)$. This completes the proof.\qed

\begin{remark}
Note that the only lossy step is the passage
from $\|g\|_1=1$ to $T\le80$ in the sum rule. For the odd-clique densities one computes directly
\begin{align}
T=\binom k2\Bigl(\frac{2}{k-1}\Bigr)^{2}=\frac{2k}{k-1}\;\le\;3\,.
\end{align}
So the true constant is smaller than $80$ by more than an order of
magnitude. Establishing $T\le3$ for all $g\in\mathcal{K}_n$ would replace $9$ by $\sqrt3$ in Theorem~\ref{thm:ncineq} and, more importantly, lower the threshold at which the inequality becomes non-vacuous from $n\ge326$ to $n\ge14$. Also, the scaling $n^{-1/2}$ is not an artifact of the method. By the lower bound of Proposition~\ref{prop:lower}, the odd cliques already achieve $\Phi=\Theta(n^{-1/2})$, so no argument can do better than the exponent obtained here.
\end{remark}

\noindent 
{\bf ACKNOWLEDGEMENT:} AC and MB thank Ronald de Wolf who, during the recent collaboration on~\cite{Chakraborty2026}, brought to their attention the connection between Hidden Matching and nonlocality~\cite{Buhrman2012}, which helped motivate the present work. JK and BC acknowledge support from University Grants Commission, India (Reference no. 231620167137 and 241620129062, respectively). MB acknowledges the financial support through the National Quantum Mission (NQM) of the Department of Science and Technology, Government of India.\\

\noindent 
{\bf Declaration on the use of AI tools:} The problem was conceived and all principal results were derived by the authors. In the course of this work the authors used Claude Opus 5 (Anthropic) as an assistive tool for editing prose, preparing figures, checking algebraic calculations, and implementing the finite enumeration underlying Theorem~\ref{thm:exact}. Certain intermediate steps were refined in the process. The authors take full responsibility for the content of the manuscript.



%

\end{document}